\documentclass{article}

\usepackage{microtype}
\usepackage{graphicx}
\usepackage{subcaption}
\usepackage{booktabs} 

\usepackage{hyperref}

\usepackage[accepted]{icml2026}

\usepackage{amsmath}
\usepackage{amssymb}
\usepackage{mathtools}
\usepackage{amsthm}
\usepackage{times}
\usepackage{xcolor}
\usepackage{bbm}

\usepackage{multirow}
\usepackage[table]{xcolor}
\usepackage[capitalize,noabbrev]{cleveref}

\theoremstyle{plain}
\newtheorem{theorem}{Theorem}[section]

\newtheorem{lemma}[theorem]{Lemma}

\theoremstyle{definition}

\theoremstyle{remark}

\usepackage[textsize=tiny]{todonotes}

\icmltitlerunning{Training-Free Oscillatory Working Memory for Salience-Driven Attention Gating}

\begin{document}

\twocolumn[
  \icmltitle{NAACA: Training-Free NeuroAuditory Attentive Cognitive Architecture with Oscillatory Working Memory for Salience-Driven Attention Gating}



  \icmlsetsymbol{equal}{*}

  \begin{icmlauthorlist}
    \icmlauthor{Zhongju Yuan}{ugent}
    \icmlauthor{Geraint Wiggins}{vub,qmul}
    \icmlauthor{Dick Botteldooren}{ugent}
  \end{icmlauthorlist}

  \icmlaffiliation{ugent}{WAVES Research Group, Ghent University, Gent, Belgium}
  \icmlaffiliation{vub}{AI Lab, Vrije Universiteit Brussel, Brussel, Belgium}
  \icmlaffiliation{qmul}{EECS, Queen Mary University of London, London, UK}

  \icmlcorrespondingauthor{Zhongju Yuan}{zhongju.yuan@ugent.be}

  \icmlkeywords{Machine Learning, ICML}

  \vskip 0.3in
]



\printAffiliationsAndNotice{}  


\begin{abstract}
Audio provides critical situational cues, yet current Audio Language Models (ALMs) face an attention bottleneck in long-form recordings where dominant background patterns can dilute rare, salient events. We introduce NAACA, a training-free \textbf{N}euro\textbf{A}uditory \textbf{A}ttentive \textbf{C}ognitive \textbf{A}rchitecture that reframes attention allocation as an auditory salience filtering problem. At its core is OWM, a neuro-inspired \textbf{O}scillatory \textbf{W}orking \textbf{M}emory that maintains stable attractor-like states and triggers higher-cognition ALM processing only when adaptive energy fluctuations signal perceptual salience, triggering higher-level reasoning. On XD-Violence, NAACA improves AudioQwen’s average precision (AP) from 53.50\% to 70.60\% while reducing unnecessary ALM invocations. Furthermore, qualitative case studies on the Urban Soundscapes of the World (USoW) dataset show that OWM captures novel events and subcategory shifts while remaining robust to transient pauses and ambient urban noise. The source code is available at \url{https://github.com/zjyuan1208/NAACA-Oscillatory-Working-Memory}.
\end{abstract}

\section{Introduction}
Audio provides critical situational cues when vision is degraded or unavailable, enabling detection of abnormal events in public spaces such as distress calls, fights, or sudden crowd escalation, while also supporting acoustic monitoring in natural environments, where shifts in soundscapes can indicate biodiversity changes. Recent Audio Language Models (ALMs) have advanced audio understanding beyond narrow classification, but face a critical pre-attentive bottleneck in real-time long-form analysis. Unlike biological hearing, which instinctively filters stable background soundscapes to prioritize salient stimuli, transformer-based ALMs allocate attention inefficiently over long contexts. When applied to long audio streams, ALM based inference can miss rare but critical events that occur late in the recording, because model focus is dominated by earlier or persistent background patterns, an effect we refer to as attention dilution. Fig.~\ref{fig:failure_case} illustrates this failure mode on a representative example. While exhaustive processing of short segments could recover these scene subsets, the computational cost is prohibitive for continuous monitoring. Consequently, a fundamental trade-off arises between perceptual saliency recall and computational efficiency.

\begin{figure}[h]
    \centering
    \includegraphics[width=\linewidth, trim=5 0 30 5, clip]{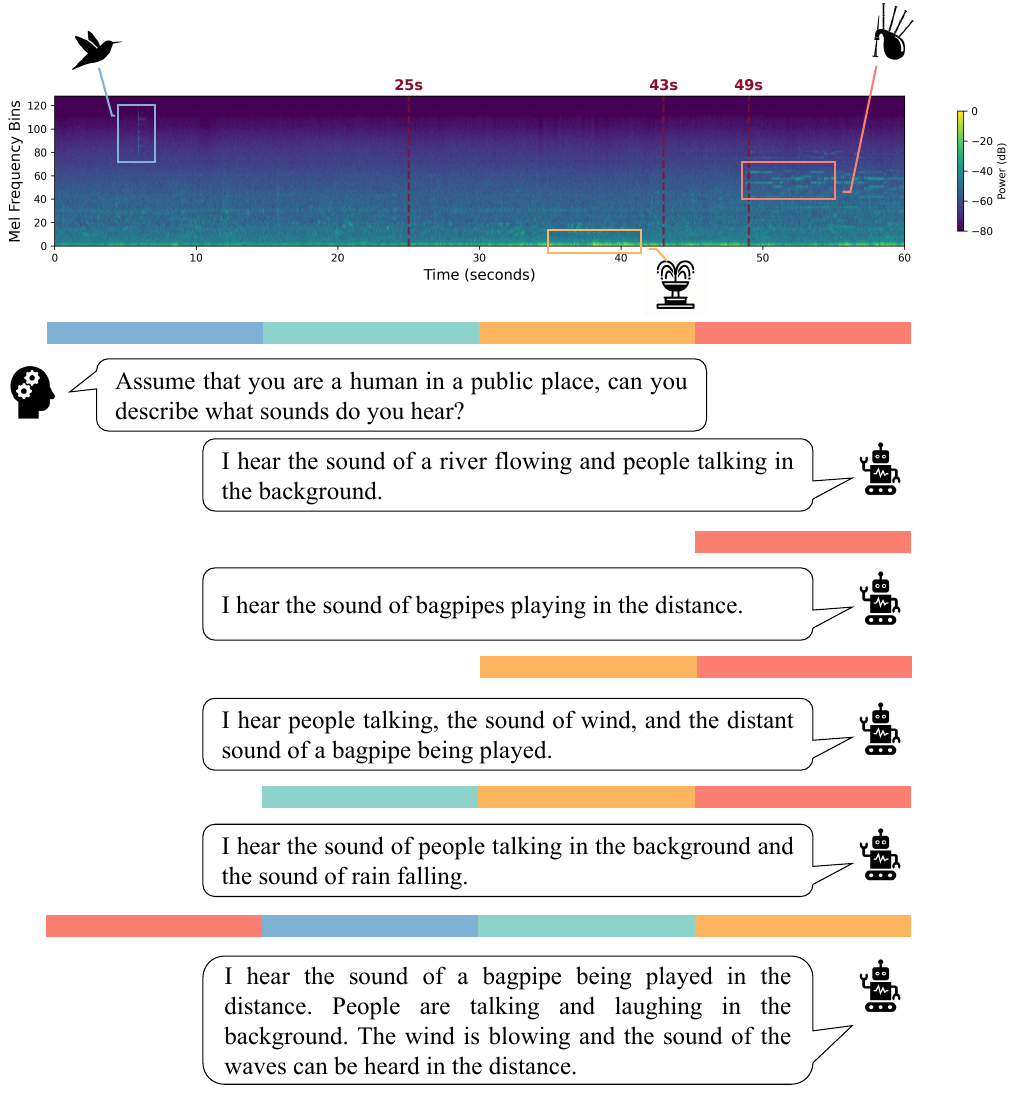} 
    \caption{\textbf{ALM attention failure and context limitations in long-form audio.} \textbf{Top:} Mel-spectrogram of sample R0056 (USoW) with three salient scenes: birdsong (blue), increased fountain noise (yellow), and bagpipe onset (red). \textbf{Middle:} Direct inference on the full 60s clip (partitioned into 15s segments) omits the terminal bagpipe event, illustrating context-length bottleneck. \textbf{Bottom:} Varying context lengths and ordering. Trailing windows (15s--45s) recover different scene subsets; moving the final 15s to the front (reordered 60s) recovers the bagpipe, confirming temporal-attention decay in the Audio-LLM.}
    \label{fig:failure_case}
\end{figure}


We address this trade-off by drawing inspiration from cognitive neuroscience, where selective attention mechanisms prioritize salient stimuli while filtering stable backgrounds. In auditory neuroscience, salience refers to unexpected or behaviorally relevant stimuli \cite{corbetta2002control, kayser2005mechanisms}. Crucially, salience is not loudness: it corresponds to context changes requiring internal updates \cite{friston2009free}, while stable backgrounds should be maintained without consuming attention \cite{winkler2009modeling}. Salient events are thus those warranting increased ALM focus.

Building on this salience framework, we reformulate long-form processing as Salience-Driven Attention Gating, an online segment selection policy that routes salient audio regions to the ALM. Rather than processing every window, we detect when the auditory pattern shifts in a salient manner and selectively invoke ALM reasoning only for those segments. This approach is critical for deployment, where offline training is prohibitive and data streams are unlabeled and non-stationary \cite{wan2024online,chan2025online}. Existing approaches, statistical detectors \cite{rabanser2019failing,chan2025online} and representation-based methods \cite{wan2024online}, require long-term historical data and substantial overhead.

To avoid the long-term historical data and substantial overhead required by existing methods, we draw inspiration from biological working memory. Working memory is supported by attractor-like neural states enabling stable maintenance while remaining sensitive to deviations requiring updates \cite{brennan2023attractor}. Oscillatory dynamics link to attention and memory control, suggesting salience emerges from state transitions rather than learned classification \cite{lundqvist2018gamma}.


Concretely, we propose \textbf{OWM} (\textbf{O}scillatory \textbf{W}orking \textbf{M}emory), a bio-inspired module integrating frequency-selective oscillatory inputs from pretrained audio representations, maintaining stable internal states as attractor-like memory items, and identifying saliency by comparing energy fluctuations to an adaptive threshold. Building on OWM, we propose \textbf{NAACA}, a NeuroAuditory Attentive Architecture performing training-free online saliency gating for ALMs. As shown in Fig.~\ref{fig:overview}, we process audio in sliding windows and convert each window into auditory object probability sequences, which drive frequency selective oscillatory inputs on the OWM grids. OWM maintains stable representations and triggers attentional focus when energy exceeds an adaptive threshold; detected segments are forwarded to an ALM for semantic interpretation. This enables operation without offline training while improving long-context reasoning by prioritizing salient transitions. On XD-Violence, our gating improves AudioQwen for 17.10 percentage points in AP while substantially reducing ALM invocations.

Our contributions are:
\begin{itemize}
  \setlength{\itemsep}{0pt}
    \item \textbf{Neuroscience-Inspired Gating:} We formalize the ALM accuracy--cost trade-off and address it via salience-driven gating.
    \item \textbf{Oscillatory Working Memory:} A bio-inspired module with optimized wave dynamics for training-free, adaptive saliency filtering.
    \item \textbf{Performance:} 70.60\% AP with near 40\% cost reduction and interpretable, neurally-plausible dynamics.
\end{itemize}

\paragraph{Conflict of Interest Disclosure.}
The authors declare that they have no financial conflicts of interest related to this work.

\section{Methods}

\subsection{NeuroAuditory Attentive Cognitive Architecture (NAACA)}
\label{sec:system-overview}

\begin{figure*}
    \centering
    \includegraphics[width=\linewidth]{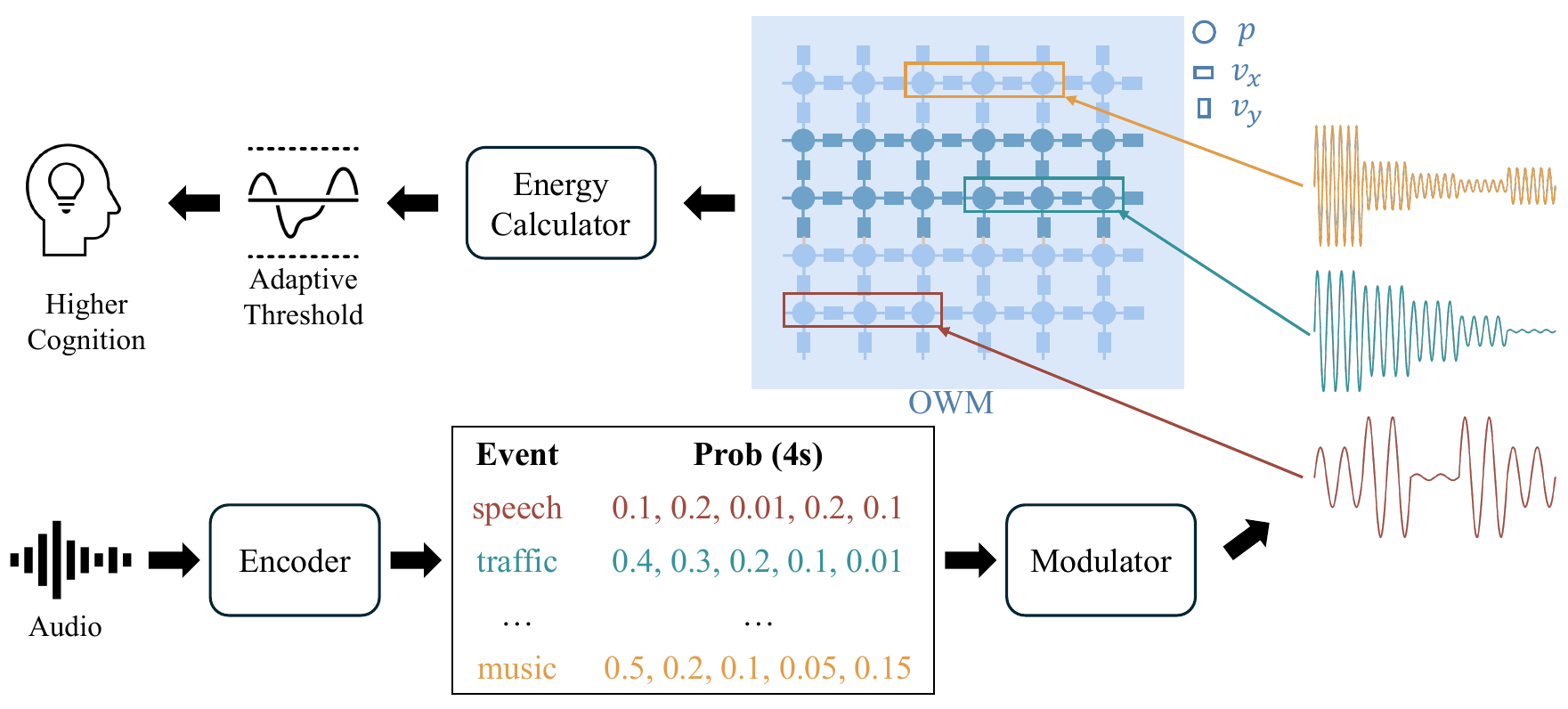} 
    \caption{\textbf{Overview of the NAACA (\textbf{N}euro\textbf{A}uditory \textbf{A}ttentive \textbf{C}ognitive \textbf{A}rchitecture).} Audio is segmented into sliding windows and mapped by a pretrained encoder to auditory object probability trajectories, which drive frequency-specific oscillatory inputs on OWM grids. OWM is a 2D neural network with primary (\(p\)) and velocity (\(v\)) neurons, parameterized by wave propagation speed \(c\) and damping \(k\), where \(c\) follows a stripe-shaped binary pattern (dark/light blue). A gate-opening decision is made from OWM energy fluctuations relative to an adaptive threshold, and salient audio segments are routed to a higher-level cognition module for semantic interpretation.}

    \label{fig:overview}
\end{figure*}

Our framework follows a multi-stage processing pipeline (Fig.~\ref{fig:overview} and Algorithm~\ref{alg:OWM_drift_detection}). Incoming audio streams are segmented into short, overlapping windows \(\mathbf{x}_t\), which are encoded into auditory object category probability vectors 
\(\mathbf{p}_t = \mathrm{Enc}(\mathbf{x}_t)\) by a pretrained encoder \(\mathrm{Enc}(\cdot)\). 

These probabilities are transformed into oscillatory drive signals through a 
predefined modulator. Specifically, each probability 
dimension is assigned a unique carrier frequency, represented as a sine wave, 
while the corresponding probability value modulates its amplitude. As illustrated 
in Fig.~\ref{fig:overview}, \(\mathrm{Enc}(\cdot)\) produces sound-category probability vectors across sliding windows, 
which are then mapped to their respective sinusoidal drive signals at distinct frequencies. Formally, the modulated source term for auditory object category \(i\) is
\begin{equation}
\label{eq:parcel_source_main}
\begin{split}
S_i(x,t) &= a_i(t)\,\sin(\omega_i t)\,\mathbf{1}_{\Omega_i}(x), \\
&\quad \omega_i = 2\pi f_i, \;\; a_i(t) \in [0,1],
\end{split}
\end{equation}
where \(a_i(t)\) denotes the instantaneous amplitude given by the encoder probability for category \(i\) at time \(t\),
\(f_i\) is the carrier frequency assigned to category \(i\)((assignment details in Appendix~\ref{app:freq_assignment})), and \(\omega_i = 2\pi f_i\).
\(\Omega_i \subset \{1,\ldots,G\}\times\{1,\ldots,G\}\) is the spatial parcel on a \(G\times G\) OWM lattice to which the
\(i\)-th probability dimension \(\mathbf{p}_t[i]\) injects an oscillatory drive, and \(\mathbf{1}_{\Omega_i}(x)\) restricts the forcing to that region.
Here \(x\) abbreviates the 2D lattice coordinate \((x,y)\); Eq.~\eqref{eq:parcel_source_main} is the compact notation of the full coordinate-wise form in Eq.~\eqref{eq:appendix_drive_signal}.
Importantly, Eq.~\eqref{eq:parcel_source_main} only defines the \emph{input}; wave-like propagation is produced by the OWM
recurrent dynamics introduced next, where spatial coupling terms (via \(\nabla p\) and \(\nabla\!\cdot v\)) allow activity induced in \(\Omega_i\)
to spread and interact across the grid under damping and wave propagation speed parameters. 
In this driven–damped setting, \(\Omega_i\) serves as a consistent locus of excitation and can yield \emph{attractor-like} category-specific
spatiotemporal patterns (e.g., resonant responses) determined by the OWM dynamics rather than by the mapping alone.

We monitor changes in the OWM system energy against an adaptive saliency threshold 
$T_{\text{adapt}}$. The OWM system energy is defined as the sum of the square of each neuron value at each time step $\Delta t$.
The $T_{\text{adapt}}$ is computed using an energy-based approach: 
\[T_{\text{adapt}} = \mu + 2\sigma\bigl(1 + \alpha \cdot \text{trend}\bigr),\]
where $\mu$ and $\sigma$ are the running mean and standard deviation of energy-derived 
drift metrics over a sliding window of $W=20$ samples, and the trend factor adjusts 
for temporal patterns in the data. The final detection decision employs persistence filtering to ensure robust salience detection while minimizing false positives, after which the detected segments are forwarded to a higher-level cognitive module for semantic interpretation. The whole procession is shown in Algorithm~\ref{alg:adaptive_threshold}.



\subsection{OWM Formulation}
\label{sec:OWM-formulation}

OWM is a 2D recurrent field model defined on a $G\times G$ lattice, where $G$ denotes the grid resolution. Its hidden state at time $t$ consists of a primary pressure-like field $p(x,y,t)$ and a velocity-like field $\mathbf{v}(x,y,t)=(v_x(x,y,t),v_y(x,y,t))$ at lattice coordinate $(x,y)$. The pressure field stores the current auditory memory state, while $\mathbf{v}$ mediates directional flow between neighboring locations. The update combines (i) temporal recurrence through the previous state for memory and (ii) spatial recurrence through a structured operator $\mathcal{A}(\cdot)$ for lateral propagation. The wave propagation speed $c(x,y)$ and damping coefficients $k^{p}(x,y)$ and $k^{v}(x,y)$ control, respectively, spatial coupling and dissipation. This bio-inspired design is analogous to membrane-potential storage ($p$) and axonal/dendritic transport ($\mathbf{v}$) in cortical sheet.

In our implementation, the pretrained PANN encoder outputs $C=527$ category probabilities. Category $i\in\{0,\ldots,C-1\}$ is assigned a carrier frequency $f_i$ and spatial parcel $\Omega_i$ on the $64\times64$ lattice; its probability $a_i(t)$ drives $\Omega_i$ through Eq.~\eqref{eq:parcel_source_main}. We set
\begin{equation}
\begin{aligned}
f_i &= f_{\min}+i\frac{f_{\max}-f_{\min}}{C-1},\\
c_i &= \frac{\tan(\pi f_i\Delta t)\sqrt{(1+\Delta t k^p)(1+\Delta t k^v)}}{\Delta t\sqrt{2}},
\end{aligned}
\label{eq:main_wave_speed_formula}
\end{equation}
where $f_{\min}=51$ Hz, $f_{\max}=1200$ Hz, $\Delta t=0.01$ s, $dx=1$, and $k^p=k^v=10$. The spatial speed field is $c(x,y)=\sum_i c_i\mathbf{1}_{\Omega_i}(x,y)$, with $\mathbf{1}_{\Omega_i}$ denoting the parcel indicator; Appendix~\ref{app:freq_assignment} gives the full derivation.

The parcel allocation is deterministic rather than learned. We enumerate the $G^2=4096$ lattice cells in row-major order and divide them across $C=527$ PANN categories as evenly as possible: the first $G^2\bmod C=407$ parcels contain eight cells, and the remaining parcels contain seven cells. This mapping gives each category a stable injection locus while avoiding any target-task optimization. Although the row-major assignment is simple, salience detection is computed from the global energy $E(t)$ rather than from local category neighborhoods alone; therefore any substantial change in the encoder probability vector can produce a global energy transient. The striped wave-speed field further amplifies such transients through modal coupling, so the method does not require a learned semantic geometry. If a different encoder with $C'$ output dimensions is used, only the deterministic frequency and parcel assignment needs to be recomputed.

\subsubsection{Constructing the structured spatial operator $\mathcal{A}(\cdot)$ of OWM}
\label{sec:math-foundations}

\paragraph{Wave System Foundations.}
We consider the OWM as a two-dimensional oscillatory system governed by a damped wave equation in the first-order velocity–pressure formulation:
{\setlength{\abovedisplayskip}{4pt}
\setlength{\belowdisplayskip}{4pt}
\begin{equation}\label{eq:wave_eq}
\begin{aligned}
\frac{\partial p}{\partial t} + k^{p}(x,y)\, p &= -c^2(x,y)\,\nabla \cdot \mathbf{v} + S(x,y,t), \\
\frac{\partial \mathbf{v}}{\partial t} + k^{v}(x,y)\, \mathbf{v} &= -\nabla p,
\end{aligned}
\end{equation}}
where \(c(x,y)\) is the spatially varying wave propagation speed (time-independent), \(k^{p}(x,y)\) and \(k^{v}(x,y)\) are pressure and velocity damping coefficients (time-independent), and \(S(x,y,t)\) is given by Eq.~\ref{eq:parcel_source_main}, and $\nabla = [\partial/\partial x, \partial/\partial y]^T$ is the gradient operator, and \(\nabla \cdot \mathbf{v} = \partial v_x/\partial x + \partial v_y/\partial y\) is the divergence operator.

Discretizing Eqs.~\ref{eq:wave_eq} over a two-dimensional lattice yields:
\begin{equation}\label{eq:disc_wave_eq}
  \small
\begin{aligned}
p(x,y,t+\Delta t) &= \bigl[1 - \Delta t\,k^{p}(x,y)\bigr]\,p(x,y,t) \\
&\quad - \Delta t\,c^{2}(x,y)\,\nabla \cdot \mathbf{v}(x,y,t) + \Delta t\,S(x,y,t), \\
\mathbf{v}(x,y,t+\Delta t) &= \bigl[1 - \Delta t\,k^{v}(x,y)\bigr]\,\mathbf{v}(x,y,t) - \Delta t\,\nabla p(x,y,t).
\end{aligned}
\end{equation}
where \(\Delta t\) is the time step.


This discrete update system exhibits frequency-selective response through spatial phase differentiation, as characterized by Theorem~\ref{thm:phase_response}.

\begin{theorem}[Frequency-Selective Phase Response]
\label{thm:phase_response}
Under multi-frequency excitation $S(x,y,t) = \sum_i S_i(x,y,t)$ with components at frequencies $\{\omega_i\}$, the steady-state pressure field exhibits frequency-dependent phase delay:
\begin{equation}
\label{eq:phase_response}
p(x,y,t) \approx \sum_i A_i(x,y) \, S_i(x,y, t - \tau_i(x,y)),
\end{equation}
where the phase delay is determined by input-resonance mismatch:
\begin{equation}
\label{eq:phase_delay}
\tau_i(x,y) \propto \arctan\left(\frac{\omega_i - \omega_{\mathrm{res}}(x,y)}{k^p(x,y) + k^v(x,y)}\right),
\end{equation}
with $\omega_{\mathrm{res}}(x,y)$ the local resonance frequency, $A_i(x,y)$ the amplitude response, and $\tau_i(x,y)$ the temporal phase delay at location $(x,y)$ for input frequency $\omega_i$. This establishes frequency selectivity via phase: different frequencies produce distinct phase delays, enabling spatial encoding of spectral content.
\end{theorem}
\vspace{-5pt}
\textit{Proof.} See Appendix~\ref{appendix:phase_response_proof}. \qed

\textbf{System Energy Measurement.}
Different spatial locations within the OWM lattice are associated with distinct attractors corresponding to specific sound event categories, enabling the system to localize and track the perceptual saliency of diverse auditory sources. Since each event is characterized by its own modulatory input frequency, the local wave propagation speed \(c(x,y)\) must also vary across space. This spatial dependence ensures that the eigenfrequency structure reflects the diversity of sound-driven dynamics in the system. 

However, while local eigenmodes are essential for modeling event-specific resonances, pattern drift detection and memory-related computations require more than analyzing these modes in isolation. Instead, the collective behavior of the system must be captured in terms of a global state variable. To this end, we define the total system energy, which aggregates pressure and velocity contributions across the lattice. This energy-based representation not only reflects the ongoing dynamics of the OWM but also forms the key computational signal for detection and optimality analysis.

The total energy of the OWM system in discrete form is
\begin{equation}
E(t) = \frac{1}{2}\sum_{i,j} \left[p_{i,j}^2(t) + v_{x,i,j}^2(t) + v_{y,i,j}^2(t)\right], \label{eq:energy_discrete}
\end{equation}
where the terms correspond to kinetic energy due to coupling, potential energy due to stiffness. For the purpose of analyzing and designing the internal structure and parameters, we approximate the 2D lattice as a continuous medium. The energy then becomes
\begin{equation}
\textstyle
E(t) = \iint \frac{1}{2}\left[p^2(x,y,t) + v_x^2(x,y,t) + v_y^2(x,y,t)\right] dxdy, 
\label{eq:energy_continuous}
\end{equation}
which will be used in the following calculations and theorem proofs as the basis for sensitivity and optimality analyses.

\subsubsection{Topological Organization with High Sensitivity to Salience}
\label{sec:topology-sensitivity}
To analyze the role of topological organization and its sensitivity to drift, we first reformulate the governing first-order velocity–pressure system in a more compact representation in Theorem~\ref{thm:second_order_equivalence}. This reformulation exposes the effective damping and restoring mechanisms and serves as a foundation for later connecting topological behavior with energy dynamics.

\begin{theorem}[Equivalence of First-Order System to Second-Order Damped Wave Equation]
\label{thm:second_order_equivalence}
The first-order pressure-velocity system Eqs.~\ref{eq:wave_eq} is equivalent to the second-order damped wave equation
\begin{equation}
  \small
\frac{\partial^2 p}{\partial t^2} + \gamma \frac{\partial p}{\partial t} + k^{v} k^{p} \cdot p = \nabla \cdot (c^2 \nabla p) + (k^{v} \mathbf{v} + \nabla p) \cdot \nabla c^2 + \left(k^{v} + \frac{\partial}{\partial t}\right) S,
\label{eq:second_order}
\end{equation}
with effective damping coefficient $\gamma = k^{p} + k^{v}$, restoring force coefficient $\mu = k^{v} k^{p}$, and modified source term $S_{\text{eff}} = (k^{v} + \partial/\partial t) S$. The term $(k^{v} \mathbf{v} + \nabla p) \cdot \nabla c^2$ represents spatial coupling from wave speed variation.
\end{theorem}
\vspace{-14pt}
\begin{proof}
    See Appendix~\ref{appendix:Second-Order Equivalence}.
\end{proof}

\vspace{-10pt}
Building on this structural equivalence, we next derive the explicit energy evolution law, which highlights how the wave propagation speed \(c\) governs energy redistribution and thereby influences stability and drift sensitivity.

\begin{theorem}[Energy Evolution]\label{the:Energy_Evolution_Equation}
When $p(x,y,t)$ and $v_x,v_y$ are governed by Eqs.~\ref{eq:wave_eq} under periodic boundaries. The energy evolution is given by
\begin{equation}
\textstyle \frac{dE}{dt} = \int \!\! \int [ pS - k^p p^2 - k^v |\mathbf{v}|^2 - (c^2-1) p \nabla \cdot \mathbf{v} ] \, dxdy.
\label{eq:energy-evolution}
\end{equation}
\end{theorem}
\vspace{-14pt}
\begin{proof}
    See Appendix~\ref{appendix:Energy Evolution Equation}.
\end{proof}
\vspace{-10pt}
The injection term $\iint p S \, dx\, dy$ in Eq.~\eqref{eq:energy-evolution} governs detection sensitivity. To maximize sensitivity, $c(x,y)$ must create frequency-selective spatial differentiation where different auditory events drive spatially distinct activity patterns through varying local resonance frequencies~\cite{lakatos2016global}, while enabling coherent coupling between slow maintenance dynamics and fast encoding transients~\cite{lundqvist2018gamma}. A striped structure $c(x,y) = c(y)$ with Bragg-matched periodicity~\cite{kushwaha1993acoustic} achieves both objectives: coherent reflections create slow-propagating modes supporting maintenance oscillations, which interact with encoding transients via phase-dependent energy injection~\cite{lundqvist2018gamma}.

\begin{theorem}[Striped Pattern Optimality]
\label{thm:striped-bragg}
Under amplitude constraint $|\delta c^2| \leq A$, the Bragg-resonant striped square wave $\delta c^2(x,y) = A \cdot \mathrm{sgn}[\cos(2\pi q_0 y/L_y)]$ maximizes modal coupling strength $|C_{(m,n),(m,n+q_0)}|$ and spatial frequency differentiation for salience detection at mode separation $q_0$.
\end{theorem}
\begin{proof}
See Appendix~\ref{appendix:bragg-proof}.
\end{proof}
\vspace{-10pt}
The Bragg-matched structure supports multi-timescale dynamics through slow-propagating coherent modes, enabling interactions between maintenance and encoding processes analogous to cross-frequency coupling in neural working memory~\cite{lundqvist2018gamma}.

\section{Experiment}
\label{subsection:experimental_setting}
\textbf{Dataset.} We evaluate NAACA using XD-Violence~\cite{wu2020not} and Urban Soundscapes of the World (USoW)~\cite{urbansoundscapes,de2017urban}. For the multimodal XD-Violence dataset, we utilize only the audio track to benchmark our auditory drift detection. As NAACA is training-free, we evaluate exclusively on the 500-sample test set, leveraging its high-granularity, frame-level labels to assess temporal boundary precision. To complement the stylized movie audio of XD-Violence, we use USoW to simulate real-world surveillance. USoW provides unlabeled, high-quality 4-channel ambisonics of stable urban environments punctuated by rare events. This serves as our primary benchmark for qualitative analysis, demonstrating OWM’s robustness against ambient noise and its interpretability in detecting meaningful sonic shifts.

\textbf{Measurements.} We evaluate our framework using a mix of ground-truth metrics and efficiency proxies. For XD-Violence, we utilize frame-level annotations to report Average Precision (AP), assessing the global ability to distinguish violent anomalies from background noise. We also report frame-level temporal precision, defined as the fraction of OWM drift points that coincide with annotated salient event frames. To measure the practical utility of our attention-gating across both datasets, we define Computational Time Saved. This metric quantifies the reduction in inference overhead achieved by invoking the ALM only at detected drift points versus exhaustive, continuous processing. By reporting the fraction of detected drifts and visualizing the OWM's internal energy states, we demonstrate the architecture's ability to capture salient shifts while minimizing the computational burden of long-form audio understanding.


\textbf{Baselines.} We evaluate our framework against four categories: 
(i) \emph{Exhaustive ALM Inference}, an internal ablation where AudioQwen processes the full stream without gating to isolate the efficiency gains of our mechanism; 
(ii) \emph{Random 4\,s Segment Selection}, which randomly forwards the same number of 4\,s segments as NAACA to separate the benefit of shorter inputs from the benefit of OWM-based selection;
(iii) \emph{Supervised audio-only models}, including HL-Net~\cite{wu2020not} and AVadCLIP~\cite{wu2025avadclip}; and 
(iv) \emph{Video-only frameworks}, covering supervised models such as S3R~\cite{wu2022self}, VadCLIP~\cite{wu2024vadclip} and Holmes-VAU~\cite{zhang2025holmes}, as well as zero-shot approaches like TRACE~\cite{siddiqui2025traces}. 
Crucially, unlike our training-free approach, these baseline models require domain-specific fine-tuning, extra training on domain knowledge, or auxiliary attention fusion layers.

\textbf{Implementation Details.} 
Audio is processed in 4\,s windows using PANN~\cite{9229505} features (527 classes). OWM uses a $64 \times 64$ grid with $\Delta t=0.01$, $k_p=k_v=10$, carrier frequencies from 51--1200\,Hz, and striped $c(y) \in [0.1,70]$ from Theorem 2.1. AudioQwen~\cite{chu2023qwen} provides semantic interpretation. Full details in Appendix~\ref{app:freq_assignment}.

\section{Results and Discussion}

\subsection{Quantitative Evaluation on XD-Violence}

\textbf{Performance Success.}
The primary goal of our framework is to achieve competitive audio-only performance using general-purpose models without the need for expensive training from scratch or domain-specific fine-tuning. Task-specific models often incur prohibitive costs in data collection and labeling, and video-based systems are frequently limited by environmental factors such as poor lighting or occlusions; in contrast, audio-based models offer a more resilient alternative that is less sensitive to such conditions.

\begin{table}[H]
\caption{Average Precision (AP) on XD-Violence, grouped by baseline type.}
\vspace{-4pt}
\label{tab:xdv_ap}
\centering
\setlength{\tabcolsep}{3pt}
\begin{tabular}{lccccc}
\toprule
Method & Training & Zero-shot & Modality & AP (\%) \\
\midrule
Audio Qwen & \checkmark & & Audio & 53.50 \\
Random 4\,s & & \checkmark & Audio & 60.44 \\
\midrule
HL-Net$^\dagger$ & \checkmark & & Audio & 60.50 \\
AVadCLIP$^\dagger$ & \checkmark & & Audio & 52.51 \\
\midrule
S3R$^*$ & \checkmark & & Video & 80.26 \\
VadCLIP$^*$ & \checkmark & & Video & 84.51 \\
Holmes-VAU$^*$ & \checkmark & & Video & 87.68 \\
\midrule
Zero-shot CLIP$^*$ & & \checkmark & Video & 17.83 \\
TRACE$^\dagger$ & \checkmark & \checkmark & Video & 83.67 \\
\midrule
\rowcolor{gray!15}
\textbf{NAACA} & & \checkmark & \textbf{Audio} & \textbf{70.60} \\
\bottomrule
\multicolumn{5}{l}{\footnotesize $^*$Results reported in Holmes-VAU paper.} \\
\multicolumn{5}{l}{\footnotesize $^\dagger$Results reported in their respective papers.} \\
\end{tabular}
\end{table}

As shown in Table~\ref{tab:xdv_ap}, we utilize Audio Qwen as our operational baseline for classifying the audio stream. While video-only models generally outperform audio-only models on the XD-Violence task due to the explicit visual cues available in action-movie-style data, our results demonstrate that audio remains a highly informative modality when processed correctly.

Our NAACA achieves an AP of 70.60\%. This represents a substantial 17.1\% absolute improvement over the baseline Audio Qwen (53.50\% AP). To separate the effect of shorter inputs from the effect of OWM-based saliency selection, we additionally evaluate a random 4\,s segment baseline that forwards the same number of segments as NAACA. This random baseline reaches 60.44\% AP, indicating that shorter inputs account for a 6.94\% gain over exhaustive inference, while OWM selection contributes an additional 10.16\% AP. Thus, both effects are beneficial, but the selective contribution of OWM is larger. In addition, OWM drift points coincide with frame-level ground-truth timestamps at 61.1\%, showing that the selected segments are temporally aligned with annotated salient events rather than merely reducing input length.

Notably, our training-free approach significantly outperforms supervised audio-only models like AVadCLIP (52.51\% AP) and HL-Net (60.50\% AP). While video-based methods like TRACE show strong results (83.67\% AP), it is important to note that TRACE is not a pure zero-shot solution in the same sense as our framework; despite being marked as zero-shot in Table 1, it relies heavily on the training of an additional temporal encoder and a cross-attention fusion layer, which we consider a form of domain adaptation. In contrast, our framework provides a specialized, training-free solution for scenarios where visual data is unavailable or where computational resources preclude the training of auxiliary fusion modules.

\textbf{The Audio-Video Gap (Failure Case Analysis).}
Despite NAACA's competitive performance, a gap remains between audio-only frameworks and supervised video models. As illustrated in the confusion matrix in Fig.~\ref{fig:xdv_confusion_matrix}, this performance ceiling is largely attributable to inherent acoustic ambiguity and dataset constraints. The \textit{Abuse} category suffers the highest misclassification, with 45.0\% of cases identified as \textit{Fighting}. This is primarily due to the extreme scarcity of data—comprising only 11 samples—and the high semantic overlap of acoustic cues like shouting. Similarly, \textit{Shooting} exhibits a 19.0\% confusion rate with \textit{Fighting}. These events frequently co-occur in real-world scenarios, and without visual cues to identify firearms, the model defaults to the broader acoustic signature of a physical altercation. Furthermore, categories such as \textit{Explosions} and \textit{Car Accidents} rely heavily on visual markers like fire or debris for definitive classification. While 14.7\% of Fighting samples are missed entirely (appearing in the None/Miss category), only 0.7\% are confused with Explosions, though sudden, high-energy acoustic transients from both categories can be acoustically similar without visual context.
\begin{figure}[!t]
    \centering
    \includegraphics[width=\linewidth, trim=0 0 0 0, clip]{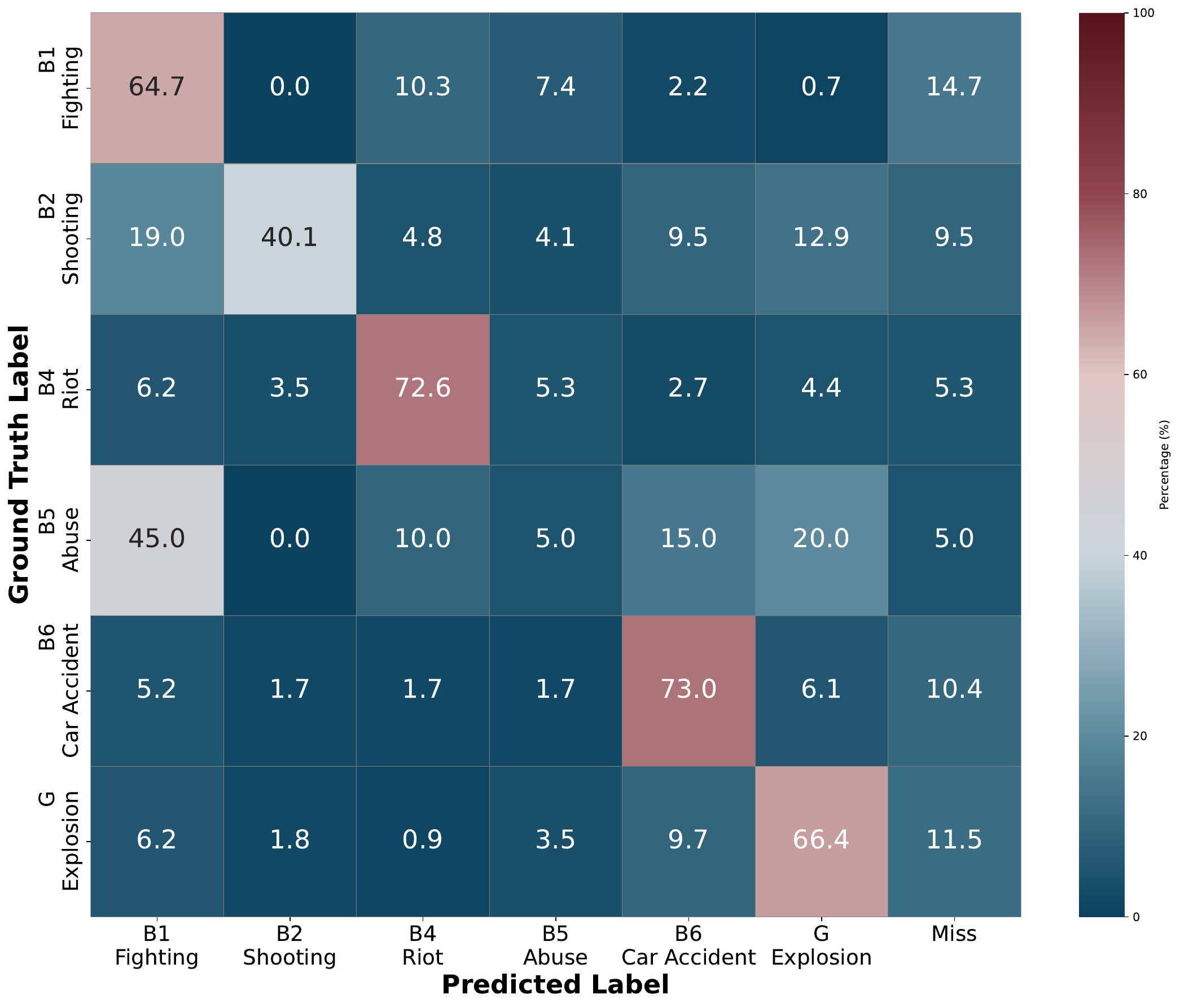}
    \vspace{-5pt}
    \caption{Confusion matrix on the XD-Violence test set audio track. Significant overlaps between \textit{Abuse}, \textit{Shooting}, and \textit{Fighting} reflect acoustic ambiguities and event co-occurrence. The misclassification of \textit{Fighting} as \textit{Explosions} highlights the reliance on visual cues for high-energy transient events.}
    \label{fig:xdv_confusion_matrix}
\end{figure}

\subsection{Qualitative Case Study on USoW}
\label{subsection:examples}

To qualitatively assess the behavior of our framework, we highlight three illustrative cases of auditory pattern drift. 
Specifically, we consider: (1) the detection of salient novel events that shifts the auditory context, 
(2) robustness to transient silences or pauses in ongoing sound streams, and 
(3) sensitivity to subcategory-level substitutions within a broader sound category (e.g., different types of musical instruments within the ``music'' class). 
These cases provide concrete insights into how the OWM detects and distinguishes different forms of drift beyond 
low-level acoustic fluctuations.

\subsubsection{Detection of Salient Novel Events}

We illustrate this case with two representative recordings in which the most salient novel event occurs near the end of the segment. In both recordings, the background soundscape is relatively stable: the Montreal recording features continuous traffic flow and bird chirping, while the Berlin recording contains crowd conversations, birds, and fountain noise. This provides a clear contrast when a new source emerges late in the sequence, allowing us to assess how well each method captures such onsets.  

\begin{figure}[!t]
  \centering
  \begin{subfigure}{\linewidth}
      \centering
      \includegraphics[width=\linewidth, trim=0 0 390 0, clip]{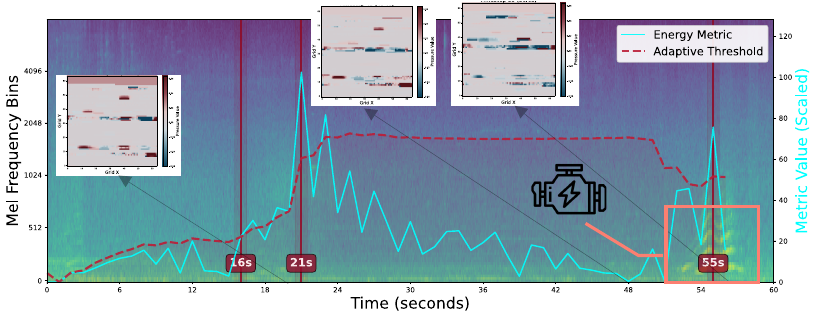}
      \subcaption{\textbf{R0002 (Place d'Armes, Montreal)}: Car engine onset at 53\,s; OWM $p$-field activation near (30, 52).}
      \label{fig:R0002}
  \end{subfigure}

  \begin{subfigure}{\linewidth}
    \centering
    \includegraphics[width=\linewidth,height=0.38\textheight,keepaspectratio,  trim=0 0 390 0, clip]{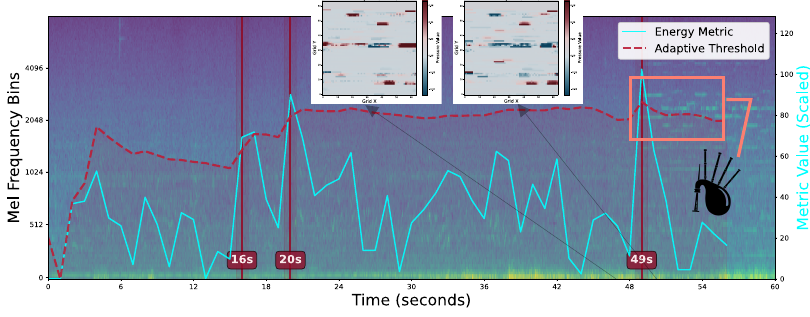}
    \subcaption{\textbf{R0056 (Alexanderplatz, Berlin)}: Bagpipe onset at 49\,s; OWM $p$-field activation near (25, 22).}
    \label{fig:R0056}
  \end{subfigure}
  \vspace{-18pt}
  \caption{\textbf{OWM detection of novel events.} Mel-spectrograms with car engine (a) and bagpipe (b) onsets. OWM outputs (cyan = energy change, red = adaptive threshold, $p$-field states shown). Vertical dashed lines mark detected drifts.}
  \label{fig:novel_sound_event}
\end{figure}

As shown in Fig.~\ref{fig:novel_sound_event}, OWM detects late-arriving novel events with high precision. The $p$-field concentrates energy near the true change point (e.g., Example~R0002, Fig.~\ref{fig:R0002}), distinguishing genuine onsets from background variability.

\subsubsection{Robustness to Transient Pauses}
We next consider recordings where the salient sound events include natural pauses or interruptions. Such cases are challenging because detectors may mistake short gaps within an ongoing event for the onset of new events. The first recording (Fig.~\ref{fig:R0037}) was captured in a dense traffic environment, while the second (Fig.~\ref{fig:R0016}) comes from a lively public square during a music festival. In both recordings, prominent sources exhibit intermittent activity, providing a useful testbed for evaluating robustness to transient pauses. 
\begin{figure}[!t]
  \centering
  \begin{subfigure}{\linewidth}
    \centering
    \includegraphics[width=\linewidth,height=0.38\textheight,keepaspectratio, trim=0 0 390 0, clip]{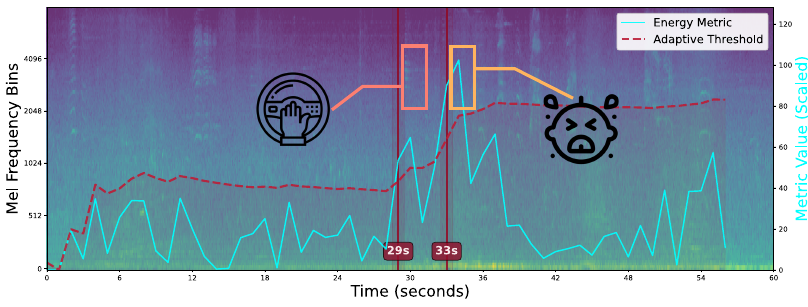}
    \subcaption{\textbf{R0037 (Johnston Road, Hong Kong)}: A baby cry with short pauses. OWM registers one event at 33\,s.}
    \label{fig:R0037}
  \end{subfigure}

  \begin{subfigure}{\linewidth}
    \centering
    \includegraphics[width=\linewidth,height=0.38\textheight,keepaspectratio, trim=0 0 390 0, clip]{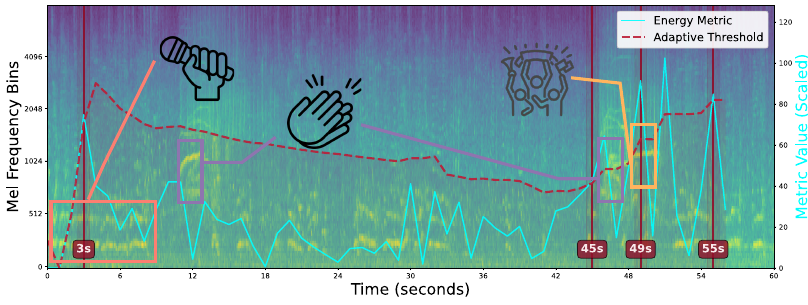}
    \subcaption{\textbf{R0016 (Quincy Market, Boston)}: Festival scene with speech followed by applause. OWM yields a single detection.}
    \label{fig:R0016}
  \end{subfigure}
  \vspace{-18pt}
  \caption{\textbf{OWM is robust to transient pauses.} Mel-spectrogram of OWM output. Vertical dashed lines mark detected drifts: cyan indicates energy change, and red indicates the adaptive threshold.}
  \label{fig:robust_to_transient_pause}
\end{figure}

As shown in Fig.~\ref{fig:robust_to_transient_pause}, OWM successfully identifies the main events while maintaining stability through spectral fluctuations. In Example~R0016, OWM effectively consolidates short silences within the same salient event. This demonstrates its ability to maintain a stable event representation and avoid over-segmentation despite transient pauses in the acoustic signal.

\subsubsection{Subcategory Transition Sensitivity}

We further analyze cases where acoustic variation arises within a subcategory of ongoing sounds rather than from the introduction of a completely new source. In this example (Fig.~\ref{fig:sensitivity_to_subcategory}), recorded in Square Phillips (Montreal), the background consists of street traffic mixed with music containing two main instruments: a hi-hat at higher frequencies and a slower kick drum. The interplay of these instruments produces several subcategory-level shifts.  

Specifically, the hi-hat drops out at 21\,s, leaving only the kick drum; it reappears at 32\,s and pauses again at 38\,s. OWM successfully detects each of these subtle transitions, capturing the precise moments of instrument entry and exit. At 42\,s, a distinct car horn emerges, which OWM identifies as a salient novel source. These results demonstrate OWM’s ability to capture fine-grained subcategory-level drift while simultaneously maintaining high sensitivity to distinct novel events.

\begin{figure}[!htbp]
  \centering
  \includegraphics[width=\linewidth, trim=0 0 390 0, clip]{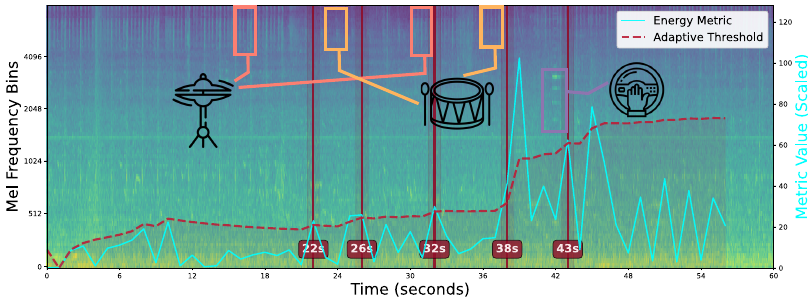}
  \caption{\textbf{OWM sensitivity to fine-grained salient transitions (Example R0010, Square Phillips, Montreal).} 
  The segment features alternating hi-hat and kick drum patterns, with a distinct car horn at 42\,s that OWM detects as a novel event alongside the subcategory drift. 
  Vertical dashed red lines mark detected drift points. The cyan curve represents energy change, the red curve shows the adaptive threshold, and detected events are highlighted at the change points.}
  \label{fig:sensitivity_to_subcategory}
\end{figure}
\vspace{-14pt}

\subsection{Spectral Analysis of \(p\)-field Dynamics via FFT}
\label{subsec:FFT_analysis}

We computed Fast Fourier Transforms (FFTs) of \(p\)-field activity, sampling internal states every second to match the sliding window stride (Subsection~\ref{subsection:experimental_setting}). With \(dt=0.01\), one second equaled 100 time steps. For each \(p\) neuron, we extracted the dominant frequency (maximal FFT amplitude) to construct frequency maps (Fig.~\ref{fig:FFT_temporal_comparison}, Examples R0016 and R0056). To suppress numerical noise, only neurons above the 75th percentile of variance were retained. Additional results appear in Appendix~\ref{appendix:Extra FFT Analysis}.

\begin{figure}[!t]
    \centering
    \begin{subfigure}[b]{0.49\textwidth}
        \centering
        \includegraphics[width=\textwidth]{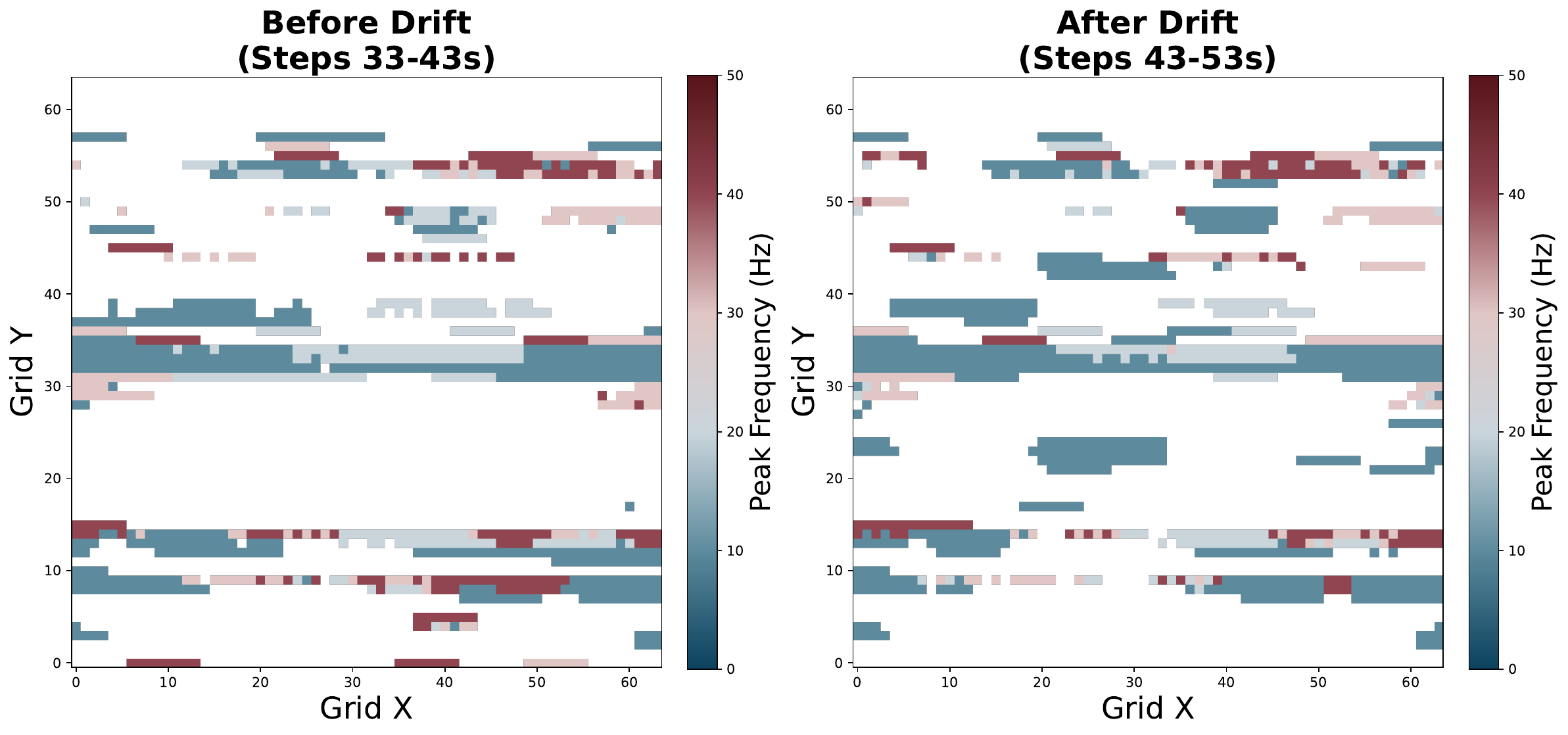}
        \subcaption{Example R0016 (drift around 43\,s)}
        \label{fig:FFT_temporal_R0016}
    \end{subfigure}
    \hfill
    \begin{subfigure}[b]{0.49\textwidth}
        \centering
        \includegraphics[width=\textwidth]{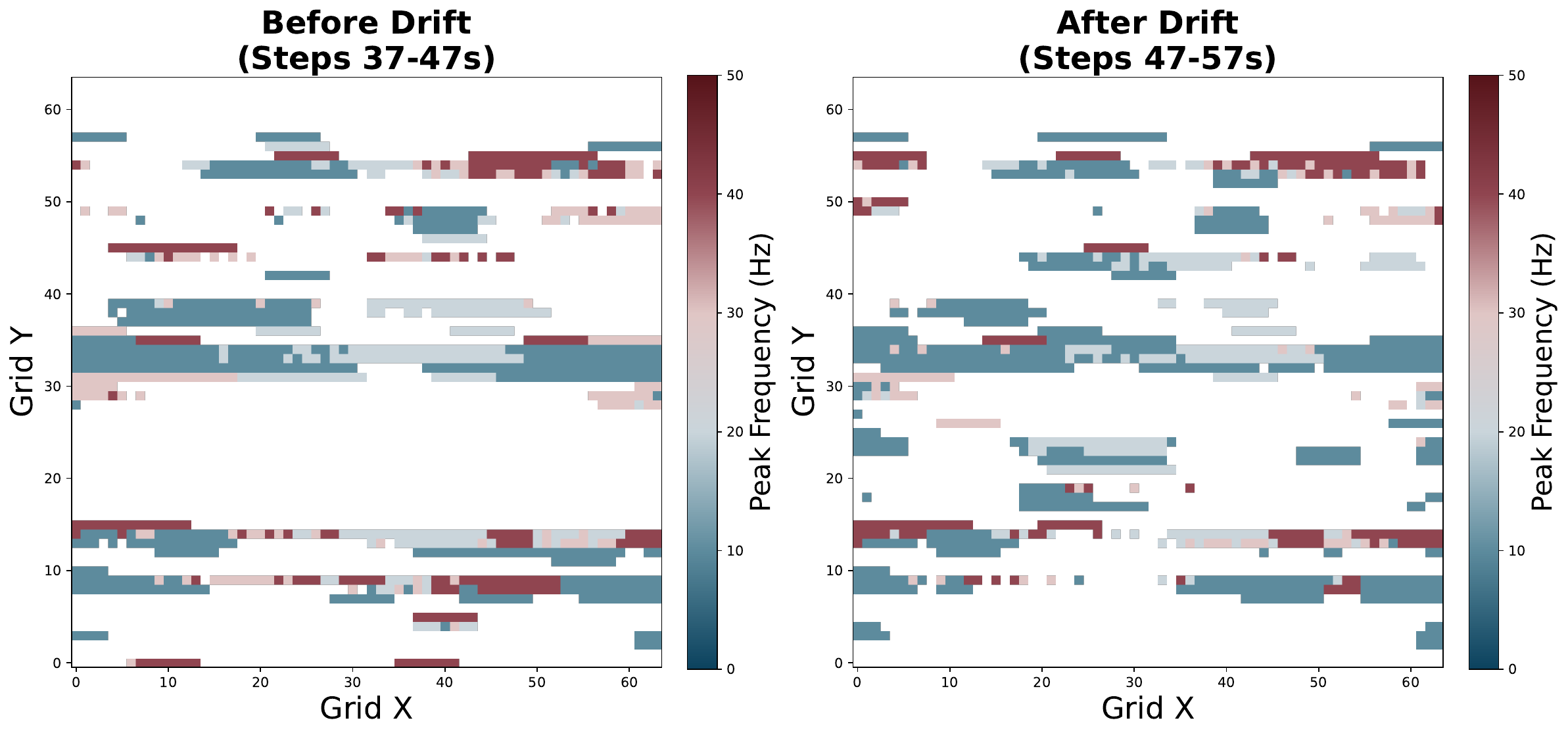}
        \subcaption{Example R0056 (drift around 47\,s)}
        \label{fig:FFT_temporal_R0056}
    \end{subfigure}
    \vspace{-15pt}
    \caption{\textbf{Temporal frequency analysis around drift detection events.} Frequency distributions in active \(p\) neurons during 10\,s before (left) and after (right) drift onset. Only neurons above the 75th percentile activity threshold are shown.}
    \label{fig:FFT_temporal_comparison}
\end{figure}
The analysis revealed spatially clustered oscillatory activity rather than uniform grid activation, with dominant frequencies limited to 0--50\,Hz by the Nyquist bound (\(dt=0.01\) s). Distinct frequency bands aligned with canonical neural regimes: during stable background periods, both examples exhibited sustained \(\beta\)-band activity (15--30\,Hz), consistent with its theorized role in working memory maintenance \cite{lundqvist2018gamma}. After drift onset, Example R0016 shifted toward \(\gamma\)-band activity (30--50\,Hz), reflecting rapid encoding of salient auditory input (applause, cheering). In Example R0056, the post-drift segment with emerging bagpipe sounds showed subsets of neurons entering the \(\gamma\) range, while pronounced \(\beta\)-band oscillations persisted, suggesting mixed maintenance and encoding dynamics. Spatially, both examples exhibited a redistribution of activity: channels along the upper grid boundary (Y=1 row), previously speech-related, showed reduced activity post-drift, whereas deeper clusters ($Y \approx 20$) became strongly engaged, consistent with recruiting new resources for encoding applause, cheering, and musical instruments. These patterns underscore that OWM reallocates oscillatory dynamics to novel salient sources rather than sustaining prior speech inputs. Supplementary animations show the real-time evolution of p-field activity for each example.

In contrast, \(\alpha\)-band activity (8--12\,Hz) remained weak with no systematic changes across drift events, precluding confirmation of links between elevated \(\alpha\) power and attentional lapses \cite{lakatos2016global, kasten2024opposing}. Similarly, \(\theta\)-band activity (4--8\,Hz) was sparse and failed to form robust clusters, despite prior reports of \(\theta\) entrainment supporting auditory working memory \cite{albouy2017selective, bonetti2024spatiotemporal}.

In sum, the OWM reallocates frequency-specific oscillatory clusters, similar to cortical $\gamma$-band encoding and $\beta$-band maintenance, to represent salient auditory transitions within its bio-physical constraints.

\subsection{Cognitive Resource Efficiency}
\begin{figure}[!t]
  \centering
  \includegraphics[width=0.85\linewidth]{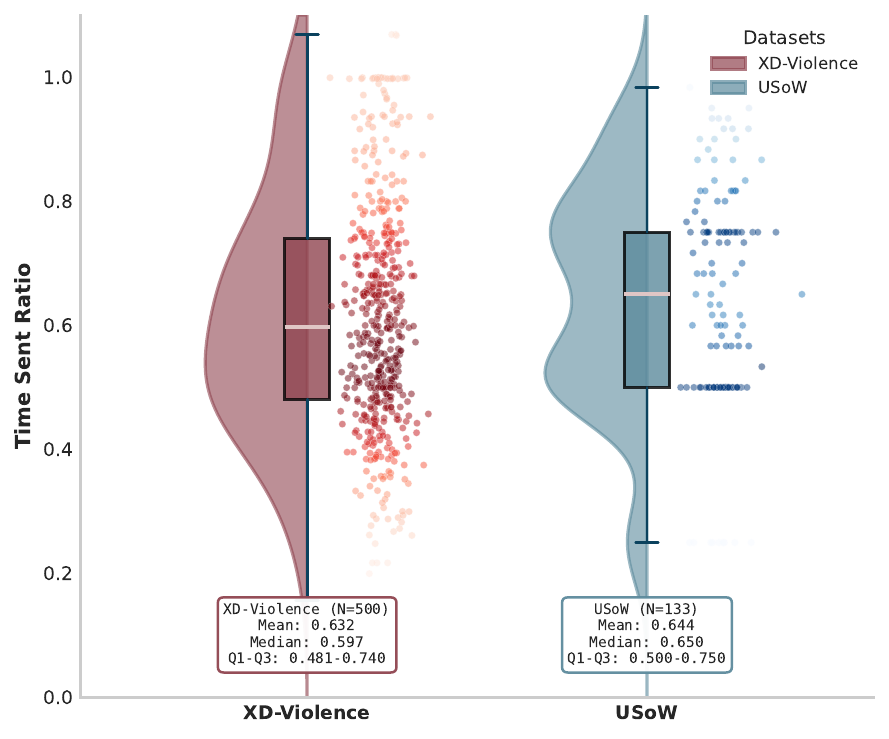}
  \vspace{-5pt}
  \caption{\textbf{Time sent ratios for XD-Violence and USoW datasets.} Violin plots with box plots and scatter points show the fraction of audio forwarded to the ALM after OWM drift detection. Both datasets exhibit similar distributions (medians: 0.597 and 0.650), demonstrating that NAACA consistently processes only ~60\% of audio duration, substantially reducing computational cost while preserving detection accuracy.}
  \label{fig:time_sent_to_qwen}
\end{figure}
Deploying ALMs for long-form audio requires balancing accuracy with computational cost. While exhaustive sliding-window processing addresses attention dilution, it requires 57 ALM invocations per 60\,s clip, which is prohibitive for real-time use. NAACA optimizes this via OWM-based gating, forwarding only salient segments. We measure efficiency using the Time Sent Ratio.

As shown in Fig.~\ref{fig:time_sent_to_qwen}, median ratios for XD-Violence and USoW are 0.597 and 0.650, respectively. This signifies a 40\% reduction in ALM invocations (from ${\sim}$57 to ${\sim}$34). Notably, this cost reduction is paired with a 17.1\% improvement in AP (70.60\% vs. 53.50\% baseline). By focusing on drift-enriched windows, OWM provides the ALM with contexts where transitions are concentrated, shifting the Pareto frontier of accuracy and efficiency.

Scatter points in Fig.~\ref{fig:time_sent_to_qwen} demonstrate that NAACA’s cost scales adaptively: ratios near 1.0 indicate high acoustic complexity, while stable soundscapes yield ratios of 0.3--0.4. This adaptive allocation spends the computational budget proportionally to information density. Despite different acoustic characteristics, OWM maintains stable gating across both datasets, positioning NAACA as a practical solution for resource-constrained monitoring.

For streaming deployment, the gating decision is made online before invoking the ALM, rather than by post-hoc pruning after full-clip inference. The algorithmic latency is mainly set by the 1s encoder stride, with the persistence filter adding only a small delay in the worst case; the subsequent ALM runtime is model- and hardware-dependent. Thus, NAACA should be interpreted as reducing unnecessary ALM invocations and forwarded audio duration, not as changing the intrinsic inference speed of the ALM itself. The high Time Sent Ratio observed for acoustically complex clips is also desirable: the system does not impose a fixed compression budget, but preserves more context when salient transitions occur frequently.

\vspace{-10pt}

\section{Conclusion}
We introduced NAACA, a neuro-inspired framework for ALM understanding enhancement, with OWM as its core working-memory component. Our approach combines a wave-based recurrent field model with an energy-driven drift detection mechanism that adaptively reallocates attention without long-term historical data or offline pretraining. We proved that binary and striped wave-speed distributions optimize drift sensitivity, and demonstrated through urban soundscape experiments that OWM reliably captures salient novel events, resists transient pauses, and subcategory saliency shifts more effectively than similarity-based baselines. OWM's oscillatory dynamics align with cortical working memory, underscoring its biological plausibility and interpretability. These results establish OWM as a computationally efficient, neuro-inspired foundation for extending long-context reasoning in multimodal systems.

\section*{Limitations}
NAACA's performance is inherently bounded by the capability of the chosen backbone encoder and ALM; stronger pretrained models will directly yield better salience detection and semantic interpretation without any modification to OWM. The current encoder (PANN) is trained on AudioSet categories, so out-of-distribution sound events in specialised domains may be missed or misidentified. The hard-gating interface is efficient for long, information-sparse streams, but may discard boundary context that soft attention or KV-cache modulation could preserve; such variants would require white-box ALM access and are left for future work. Finally, our evaluation focuses on anomaly detection via AP and temporal precision, and future work should include SpeechIQ-style tasks such as audio question answering, instruction following, and multi-turn event summarization to assess whether OWM-gated inputs preserve the context needed for deeper reasoning.

\section*{Acknowledgment}
This work was supported by the Special Research Fund (BOF) of Ghent University under Grant BOF/24J/2021/246, and by the Flemish Government through the Flanders AI Research programme (Onderzoeksprogramma AI Vlaanderen programme).

\section*{Impact Statement}


This paper presents work whose goal is to advance the field of Machine
Learning. There are many potential societal consequences of our work, none
which we feel must be specifically highlighted here.


\nocite{langley00}

\bibliography{example_paper}
\bibliographystyle{icml2026}

\newpage
\appendix
\onecolumn
\section*{Appendix}
\renewcommand{\thefigure}{\Alph{section}.\arabic{figure}}
\numberwithin{figure}{section}
\renewcommand{\thetable}{\Alph{section}.\arabic{table}}
\numberwithin{table}{section}
\renewcommand{\theequation}{\Alph{section}.\arabic{equation}}
\numberwithin{equation}{section}
\renewcommand{\thealgorithm}{\Alph{section}.\arabic{algorithm}}
\numberwithin{algorithm}{section}

\section{Related Work}
Auditory selective attention and temporal pattern drift have been studied both in neuroscience, where oscillatory dynamics are linked to memory and attentional control, and in machine learning, where attention bottlenecks in audio language models (ALMs) constrain long-sequence processing. Insights from neuroscience on tracking behaviorally relevant inputs and adapting to changes over time motivate analogous strategies for ALMs operating on long audio streams. 
In our setting, we use \emph{temporal drift} to denote a shift in the audio stream's statistics or in its task-relevant representations (e.g., embedding or class-probability trajectories) across time. 
This is closely related to \emph{concept drift} in machine learning, which refers to changes in the underlying data-generating process that alter the mapping between inputs and the semantics or labels of interest. Thus, detecting temporal drift in representations provides an online, label-free proxy for potential concept drift and can trigger reallocation of limited inference resources.

\subsection{Oscillatory Dynamics in Memory and Attention}
One prevailing hypothesis is that working memory is supported by discrete neural activity configurations, often described as attractor states~\cite{brennan2023attractor}. However, attractor dynamics alone may fail to explain working memory across multiple timescales, motivating models that couple attractor-like dynamics with activity-dependent plasticity~\cite{brennan2023attractor}. Within this dynamical view, neural oscillations provide a complementary mechanism for coordinating information maintenance and prioritization over time. Neural oscillations in specific frequency bands have been shown to play a central role in memory, selective attention, and sensitivity to temporal drift. For example, \cite{lundqvist2018gamma} demonstrated that working memory tasks involve non-stationary dynamics, with gamma bursts during encoding and beta bursts during maintenance. Similarly, selective entrainment of theta oscillations has been shown to enhance auditory working memory performance~\cite{albouy2017selective, bonetti2024spatiotemporal}. In contrast, high alpha-band activity has been associated with increased error rates and reduced auditory attention~\cite{lakatos2016global, kasten2024opposing}.  

\subsection{Attention Limitations and Concept Drift Detection}
Inspired by these oscillatory mechanisms, one can view attention in ALMs as a resource that must be selectively and dynamically allocated in response to temporal pattern drift rather than distributed uniformly across the input. Although recent advances in ALMs have enabled significant progress in audio understanding, long-form reasoning remains limited by restricted attention span, motivating efforts to extend context length~\cite{wu2023hiformer, he2024ma, bai2024audiolog}. Most existing solutions require retraining or fine-tuning, which is computationally costly and inflexible. A complementary perspective is to frame attention allocation as a drift detection problem, where established methods from the machine learning literature may provide efficient and adaptive alternatives.  


In real-world audio monitoring, \emph{salient} events are inherently rare and context dependent, and many of the situations that make them important are only evident at run time. As a result, it is generally impractical to construct exhaustive supervised datasets that cover all relevant foreground events and background conditions, and to train a dedicated detector for each deployment scenario. Instead, the scope can be narrowed by operating in a task-relevant \emph{semantic} space: a high-dimensional auditory representation learned from large corpora that captures human-recognizable sound structure and supports downstream reasoning. Within this representation space, salience can be operationalized as online distributional change, motivating unsupervised drift detection methods that trigger selective processing without requiring explicit event labels.


\paragraph{High-dimensional semantic embedding challenges.}
Accordingly, our setting calls for online drift detection in semantic audio representations, such as class-probability trajectories or embedding sequences produced by pretrained audio recognition models. 
By operating in this high-dimensional representation space learned from large corpora, changes in acoustic context manifest as distributional shifts over time, providing a practical signal for unsupervised drift detection without requiring explicit event labels. Naive statistical approaches, such as tracking Euclidean distances from exponential moving averages, face the curse of dimensionality: as dimensions increase, distance metrics lose discriminative power as all points become approximately equidistant~\cite{aggarwal2001surprising, beyer1999nearest}. A key difficulty in high-dimensional drift detection is that the same change magnitude can correspond to qualitatively different shifts. 
For example, consider two consecutive class-probability vectors with identical $L_2$ distance: in one case, probability mass moves coherently between a small set of semantically related classes (e.g., \textit{gunshot} and \textit{explosion}); in another case, the same $L_2$ distance arises from small, diffuse fluctuations spread across many unrelated classes. 
A single threshold on $\|p_{t+1}-p_t\|_2$ cannot reliably distinguish these cases, motivating drift measures that exploit structure in the semantic space rather than magnitude alone.

\paragraph{Historical data and training requirements.}
Cluster-based approaches~\cite{chan2025online} assume sufficient samples per category to maintain stable cluster statistics, which is unrealistic for heterogeneous urban soundscapes where salient events (\textit{e.g.}, sirens, screams) are rare and unpredictable. MCD-DD~\cite{wan2024online} employs contrastive learning on encoder representations to compute maximum concept discrepancy, but requires maintaining reference distributions from extensive historical windows (typically $>$100 samples) and incurs significant computational overhead from pairwise similarity computation. DriftLens~\cite{greco2025unsupervised} detects drift from deep representations in real time, but assumes an \emph{offline} instantiation phase that estimates a fixed reference (baseline) embedding distribution and detection thresholds from historical/training-era data. A different setting is training-free stream monitoring, where semantic representations are extracted from pretrained audio encoders but access to encoder training data or a dedicated in-domain reference corpus may be unavailable, making offline baseline construction less straightforward.

\paragraph{Deployment constraints.}
Real-time audio monitoring, such as public safety surveillance or environmental monitoring, imposes strict constraints: (1) \textit{No offline training}: Systems must operate on new deployments without domain-specific fine-tuning; (2) \textit{No big historical buffering}: Embedded devices cannot store hours of high-dimensional embeddings; (3) \textit{Cold-start capability}: Detection must work immediately without warm-up periods; (4) \textit{Non-stationary environments}: Acoustic distributions evolve (morning traffic $\to$ evening crowds), rendering historical baselines obsolete.

Table~\ref{tab:comparison} summarizes these requirements. Our OWM-based approach is designed to satisfy them by maintaining a compact online state via oscillatory dynamics, which summarizes recent context without storing long histories of high-dimensional embeddings. When we describe NAACA as training-free, we mean that the internal network parameters of the fixed pretrained audio encoder, OWM, and ALM are not fine-tuned or updated for the target deployment; the online adaptive threshold is still computed from the streaming signal. No long historical data accumulation refers to not requiring an offline reference corpus or baseline calibration phase prior to monitoring.

\begin{table}[H]
\centering
\caption{Requirements comparison of drift detection approaches for audio-semantic streams. \textit{Online}: Processes streaming data without batch accumulation. \textit{No long historical data}: Does not require long-term buffering ($>$50 samples) of past embeddings or statistics. \textit{High-dim}: Effective on $>$500 dimensional semantic embeddings without distance concentration issues. \textit{Topology-preserving}: Maintains structured spatial relationships in semantic representation space, not just raw semantic features. The symbol $\sim$ indicates partial support or conditional availability depending on configuration parameters.}
\label{tab:comparison}
\begin{tabular}{lcccc}
\toprule
\textbf{Method} & \textbf{Online} & \textbf{No long} & \textbf{High-dim} & \textbf{Topology} \\
 & \textbf{detection} & \textbf{historical data} & \textbf{($ > $500D)} & \textbf{-preserving} \\
\midrule
Statistical (EMA/Variance) & \checkmark & \checkmark & $\times$ & $\times$ \\
MCD-DD \cite{wan2024online} & \checkmark & $\times$ & \checkmark & $\times$ \\
DriftLens \cite{greco2025unsupervised} & \texttildelow & $\times$ & \checkmark & $\times$ \\
Cluster-based \cite{chan2025online} & \checkmark & $\times$ & \texttildelow & $\times$ \\
\midrule
\textbf{OWM (ours)} & \textbf{\checkmark} & \textbf{\checkmark} & \textbf{\checkmark} & \textbf{\checkmark} \\
\bottomrule
\end{tabular}
\end{table}

\subsection{Selective Attention in Video-Language Models}

The challenge of redundant long-context processing is a central theme in recent Video Language Models (VLMs). Dispider~\cite{qian2025dispider} introduces a disentangled perception framework where a ``Decision'' module monitors coarse-grained video to proactively trigger a ``Reaction'' module only during relevant interaction moments. Similarly, Holmes-VAU~\cite{zhang2025holmes} utilizes an anomaly-focused Temporal Sampler (ATS) that leverages a lightweight detector to generate anomaly scores, allowing the VLM to selectively aggregate salient segments while filtering uninformative background.

Functionally, both the Decision module in Dispider and the ATS in Holmes-VAU serve as asynchronous gatekeepers similar to our Oscillatory Working Memory (OWM): they decouple continuous signal monitoring from expensive autoregressive inference to resolve the accuracy-cost trade-off. However, a key distinction lies in the detection mechanism. While the samplers in Dispider and Holmes-VAU typically rely on representation-based or supervised detectors that require dataset-specific training or historical data, our OWM approach operates training-free by maintaining a compact online state via oscillatory dynamics, eliminating the need for long-term historical buffering or offline calibration while achieving comparable or superior detection sensitivity.


\section{Carrier Frequency Assignment}
\label{app:freq_assignment}

\subsection{Frequency Allocation Strategy}

We assign each PANN dimension \(i \in \{0, 1, \ldots, 526\}\) a distinct carrier frequency \(f_i\) to enable frequency-multiplexed oscillatory drives within the OWM lattice. The allocation follows a deterministic procedure ensuring compatibility with the discrete wave dynamics.

\subsubsection{Linear Frequency Distribution}

Given \(C = 527\) PANN dimensions and a target frequency band \([f_{\min}, f_{\max}]\), carrier frequencies are uniformly distributed:
\begin{equation}
f_i = f_{\min} + (f_{\max} - f_{\min}) \times \frac{i}{C-1}, \quad i \in \{0, 1, \ldots, 526\}.
\end{equation}
In our implementation, \(f_{\min} = 51\) Hz and \(f_{\max} = 1200\) Hz, yielding an average frequency spacing of \(\Delta f \approx 2.18\) Hz.

\subsubsection{Frequency-to-Grid Mapping}

Each frequency \(f_i\) is mapped to a spatial parcel \(\Omega_i \subset \{0,\ldots,G-1\} \times \{0,\ldots,G-1\}\) on the \(G \times G\) lattice, where \(G=64\).

\paragraph{Grid enumeration.} All grid coordinates are enumerated in row-major order:
\begin{equation}
\mathcal{G} = \{(x,y) : x \in \{0,\ldots,G-1\}, y \in \{0,\ldots,G-1\}\},
\end{equation}
where $|\mathcal{G}| = G^2 = 4096$.

\paragraph{Parcel allocation.} The grid \(\mathcal{G}\) is partitioned into \(C\) contiguous parcels \(\{\Omega_i\}_{i=0}^{C-1}\). To distribute the \(G^2\) grid points as evenly as possible, we first compute:
\begin{equation}
    n_{\text{base}} = \left\lfloor \frac{G^2}{C} \right\rfloor = 7, \quad n_{\text{rem}} = G^2 \bmod C = 407.
\end{equation}
The first \(n_{\text{rem}}\) parcels receive one additional grid point to accommodate the remainder:
\begin{equation}
    |\Omega_i| = 
\begin{cases}
n_{\text{base}} + 1 = 8 & \text{if } i < n_{\text{rem}}, \\
n_{\text{base}} = 7 & \text{if } i \geq n_{\text{rem}}.
\end{cases}
\end{equation}
For each parcel \(i\), the boundaries are determined by cumulative allocation:
\begin{align}
\text{start}_i &= \sum_{j=0}^{i-1} |\Omega_j|, \\
\text{end}_i &= \text{start}_i + |\Omega_i|.
\end{align}
The parcel \(\Omega_i\) then consists of grid positions with linear indices in \([\text{start}_i, \text{end}_i)\):
\begin{equation}
\Omega_i = \left\{(x,y) \in \mathcal{G} : \text{start}_i \leq x \cdot G + y < \text{end}_i\right\}.
\end{equation}

Note that \(\sum_{i=0}^{C-1} |\Omega_i| < G^2\), with the remaining grid points left unassigned.


\subsubsection{wave propagation speed Computation}

For each parcel \(\Omega_i\), the local wave propagation speed \(c_i\) is computed from the target frequency \(f_i\) using the eigenfrequency relation from Theorem~\ref{thm:phase_response} Eq.~\eqref{eq:local_res_freq} in Appendix~\ref{appendix:phase_response_proof}:
\begin{equation}
c_i = \frac{\tan(\pi f_i \Delta t) \sqrt{(1 + \Delta t k_p)(1 + \Delta t k_v)}}{\Delta t \sqrt{2}},
\label{eq:wave_speed_formula}
\end{equation}
where \(\Delta t = 0.01\) is the time step, \(k_p = k_v = 10.0\) are damping coefficients, and the factor \(\sqrt{2}\) accounts for the 2D isotropic spatial discretization with \(dx = 1\).

\paragraph{Stability constraints.} To ensure numerical stability, the computed wave propagation speed is clamped:
\begin{align}
c_{\max} &= \frac{dx}{\Delta t \sqrt{2}} \sqrt{(1 + \Delta t k_p)(1 + \Delta t k_v)} \approx 77.8, \\
c_i &\leftarrow \min\left(c_i, 0.9 \times c_{\max}\right) \approx 70.0, \\
c_i &\leftarrow \max(c_i, 0.1).
\end{align}
The upper bound prevents CFL-like violations, while the lower bound ensures non-zero propagation for high-frequency modes.

\paragraph{Spatial wave propagation speed field.} The wave propagation speed field \(c(x,y)\) on the lattice is constructed by assigning \(c_i\) to all positions in parcel \(\Omega_i\):
\begin{equation}
c(x,y) = \sum_{i=0}^{526} c_i \cdot \mathbb{1}_{\Omega_i}(x,y),
\end{equation}
where \(\mathbb{1}_{\Omega_i}(x,y)\) is the indicator function. This produces a spatially-varying field with approximately binary contrast between low-frequency (small \(c_i\)) and high-frequency (large \(c_i\)) parcels.

\textbf{Nyquist Constraint.} The lower bound $f_{\min} = 51$ Hz is chosen to strictly exceed the Nyquist frequency $f_N = 1/(2\Delta t) = 50$ Hz imposed by the discrete time step $\Delta t = 0.01$ s. Setting $f_{\min} = f_N$ would cause the drive signal $\sin(2\pi f_{\min} t)$ to evaluate to zero at all discrete timesteps $t = n\Delta t$ (since $\sin(\pi n) = 0$), and would create a singularity in the wave propagation speed formula (Eq.~\ref{eq:wave_speed_formula}) where $\tan(\pi f_{\min}\Delta t) = \tan(\pi/2) \to \infty$. The choice of 51 Hz maintains the original frequency spacing of approximately 2.18 Hz while ensuring well-defined dynamics.

\subsection{Oscillatory Drive Signal Construction}

At each time step \(t\), the modulated excitation for PANN dimension \(i\) is:
\begin{equation}
S_i(x,y,t) = a_i(t) \sin(2\pi f_i t) \cdot \mathbb{1}_{\Omega_i}(x,y),
\label{eq:appendix_drive_signal}
\end{equation}
where \(a_i(t) \in [0,1]\) is the PANN probability (sigmoid-activated logit) for event category \(i\) at time \(t\). This is the full 2D coordinate form of Eq.~\eqref{eq:parcel_source_main}, where the main text suppresses the \(y\)-dependence by writing \(x\) for the lattice coordinate. The total drive signal is the superposition:
\begin{equation}
S(x,y,t) = \sum_{i=0}^{526} S_i(x,y,t).
\end{equation}

\paragraph{Batch implementation.} In the batch-aware implementation, time \(t\) is tracked independently for each sample \(b \in \{1,\ldots,B\}\):
\begin{equation}
t^{(b)} = \text{time\_steps}^{(b)} \times \Delta t,
\end{equation}
where \(\text{time\_steps}^{(b)}\) is an integer counter maintained in \texttt{\_batch\_states[time\_steps]}. The drive signal for batch \(b\) is:
\begin{equation}
S^{(b)}(x,y,t^{(b)}) = \sum_{i=0}^{526} a_i^{(b)}(t^{(b)}) \sin(2\pi f_i t^{(b)}) \cdot \mathbb{1}_{\Omega_i}(x,y).
\end{equation}

\subsection{Nyquist Constraint and Phase Encoding}

The discrete-time system with \(\Delta t = 0.01\) s imposes a Nyquist frequency:
\begin{equation}
f_{\text{Nyquist}} = \frac{1}{2\Delta t} = 50 \text{ Hz}.
\end{equation}
This constrains the internal oscillatory dynamics of the pressure field \(p(x,y,t)\) to frequencies below 50 Hz. FFT analysis (Appendix~\ref{appendix:Extra FFT Analysis}) confirms dominant \(p\)-field frequencies lie within theta (4--8 Hz), alpha (8--12 Hz), beta (13--30 Hz), and low-gamma (30--50 Hz) bands.

\paragraph{High-frequency drive encoding.} Drive frequencies \(f_i > 50\) Hz do not produce corresponding oscillations in \(p(x,y,t)\) due to the Nyquist constraint. Instead, these drives encode information through phase relationships. For integer time indices \(n = t/\Delta t\), the drive signal phase:
\begin{equation}
\phi_i(n) = 2\pi f_i n \Delta t \mod 2\pi
\end{equation}
varies uniquely for each \(f_i\), enabling discrimination between 527 semantic categories via phase-based interference patterns and energy accumulation rates, despite internal dynamics remaining below 50 Hz.

\subsection{Implementation Parameters}

The frequency assignment is fully determined by the following fixed parameters:
\begin{itemize}
\item PANN encoder output dimension: \(C = 527\)
\item Grid resolution: \(G = 64\)
\item Time step: \(\Delta t = 0.01\) s
\item Spatial step: \(dx = 1\)
\item Damping coefficients: \(k_p = k_v = 10.0\)
\item Frequency range: \([f_{\min}, f_{\max}] = [51, 1200]\) Hz
\end{itemize}
These parameters remain constant across all experiments and datasets (XD-Violence, USoW), requiring no domain-specific calibration or retraining.

\section{Detailed Algorithms}
\label{sec:detailed_algorithms}

This section provides the complete algorithmic specifications for the OWM-based drift detection framework presented in Subsection~\ref{sec:system-overview}. We present two key algorithms: the main drift detection pipeline and the adaptive threshold computation mechanism.

\subsection{Main Salience Detection Algorithm}
\label{subsec:main_algorithm}

\begin{algorithm}
\caption{OWM-based Auditory Pattern Drift Detection}
\label{alg:OWM_drift_detection}
\begin{algorithmic}[1]
\REQUIRE Audio stream $\{\mathbf{x}_t\}$, encoder $\mathrm{Enc}(\cdot)$, OWM model with spatial operator $\mathcal{A}$
\REQUIRE Persistence duration $P=3$, cooldown period $C=3$
\STATE Initialize energy calculator
\STATE Initialize adaptive threshold $T_{\text{adapt}}$ for energy metric
\STATE Initialize detection buffer $\mathcal{D} = \emptyset$, cooldown timer $t_{last} = -1$
\FOR{each time step $t$}
    \STATE $\mathbf{p}_t \leftarrow \mathrm{Enc}(\mathbf{x}_t)$ \hfill $\triangleright$ \textit{Extract event probabilities}
    \STATE Generate oscillatory inputs: $S_i(x,t) = a_i(t) \sin(\omega_i t) \mathbf{1}_{\Omega_i}(x)$
    \STATE Update OWM state with oscillatory drive signals
    
    \STATE Calculate current energy change rate from OWM dynamics \hfill $\triangleright$ \textit{Energy change rate}
    
    \STATE $T_{\text{adapt}} \leftarrow \mu + 2\sigma(1 + \alpha \cdot \text{trend})$ \hfill $\triangleright$ \textit{Update adaptive threshold}
    
    \STATE $d_{candidate} \leftarrow \mathbf{1}(\text{energy change rate} > T_{\text{adapt}})$ \hfill $\triangleright$ \textit{Energy-based detection}
    
    \STATE Add $d_{candidate}$ to detection buffer $\mathcal{D}$ \hfill $\triangleright$ \textit{Persistence filtering}
    \IF{$|\mathcal{D}| \geq P$}
        \STATE $r_{persist} \leftarrow \frac{1}{P}\sum_{i=0}^{P-1} \mathcal{D}[-P+i]$ \hfill $\triangleright$ \textit{Persistence ratio}
        \IF{$r_{persist} \geq 0.5$ AND $t - t_{last} > C$}
            \STATE \textbf{output} Drift detected at time $t$
            \STATE $t_{last} \leftarrow t$, clear $\mathcal{D}$ \hfill $\triangleright$ \textit{Reset detection state}
        \ENDIF
    \ENDIF
\ENDFOR
\end{algorithmic}
\end{algorithm}

\begin{algorithm}[!t]
\caption{Online Adaptive Threshold Calculation}
\label{alg:adaptive_threshold}
\begin{algorithmic}[1]
\REQUIRE Window size $W=20$, trend adjustment factor $\alpha=0.2$
\STATE Initialize value buffer $\mathcal{V} = \emptyset$ with maximum size $W$
\STATE \textbf{function} AdaptiveThreshold.update($v_{new}$)
\STATE Add $v_{new}$ to buffer $\mathcal{V}$
\IF{$|\mathcal{V}| < 5$} 
    \STATE $\mu \leftarrow \frac{1}{|\mathcal{V}|}\sum_{i} \mathcal{V}_i$ \hfill $\triangleright$ \textit{Insufficient data for robust statistics}
    \IF{$|\mathcal{V}| > 1$}
        \STATE $\sigma \leftarrow \sqrt{\frac{1}{|\mathcal{V}|-1}\sum_{i} (\mathcal{V}_i - \mu)^2}$
    \ELSE
        \STATE $\sigma \leftarrow 0.1$
    \ENDIF
    \STATE \textbf{return} $\mu + 1.5\sigma$ \hfill $\triangleright$ \textit{Simple threshold for bootstrap}
\ENDIF \hfill $\triangleright$ \textit{Compute baseline statistics}
\STATE $\mu \leftarrow \frac{1}{|\mathcal{V}|}\sum_{i} \mathcal{V}_i$ \hfill $\triangleright$ \textit{Running mean}
\STATE $\sigma \leftarrow \sqrt{\frac{1}{|\mathcal{V}|-1}\sum_{i} (\mathcal{V}_i - \mu)^2}$ \hfill $\triangleright$ \textit{Running std}

\IF{$|\mathcal{V}| \geq 3$} 
    \STATE $\mathbf{x} \leftarrow [0, 1, \ldots, |\mathcal{V}|-1]$ \hfill $\triangleright$ \textit{Time indices}
    \STATE $\text{slope} \leftarrow \frac{\sum_i (x_i - \bar{x})(\mathcal{V}_i - \mu)}{\sum_i (x_i - \bar{x})^2}$ \hfill $\triangleright$ \textit{Linear regression slope}
    \STATE $f_{trend} \leftarrow \frac{|\text{slope}|}{\sigma + 10^{-8}}$ \hfill $\triangleright$ \textit{Normalized trend strength}
\ELSE
    \STATE $f_{trend} \leftarrow 0$
\ENDIF

\STATE $T_{adapt} \leftarrow \mu + 2\sigma(1 + \alpha \cdot f_{trend})$ \hfill $\triangleright$ \textit{Adaptive threshold with trend adjustment}
\STATE \textbf{return} $T_{adapt}$
\end{algorithmic}
\end{algorithm}

Algorithm~\ref{alg:OWM_drift_detection} describes the complete workflow for detecting auditory pattern drift using the NAACA. The algorithm processes streaming audio input through several key stages: feature extraction via a pretrained encoder, oscillatory signal generation, OWM state updates, energy-based change detection, and persistence filtering.

The energy rate (first derivative) within the well-designed structure of the OWM system captures sudden transitions in the rate of energy change, which correspond to significant shifts in auditory patterns. The algorithm employs an adaptive threshold mechanism that automatically adjusts to the system's baseline behavior and temporal trends, eliminating the need for manual threshold tuning.

To ensure robust detection and minimize false positives, the algorithm incorporates persistence filtering that requires a minimum proportion of recent detections before confirming a drift event. Additionally, a cooldown period prevents redundant detections of the same drift event.

\subsection{Adaptive Threshold Computation}
\label{subsec:adaptive_threshold}

Algorithm~\ref{alg:adaptive_threshold} details the online computation of adaptive thresholds for energy-based drift detection. The threshold adapts to both the statistical properties of recent observations and temporal trends in the data.

The algorithm maintains a sliding window of recent metric values and computes the threshold as $T = \mu + 2\sigma(1 + \alpha \cdot f_{\text{trend}})$, where $\mu$ and $\sigma$ are the running mean and standard deviation, $\alpha$ is the trend adjustment factor, and $f_{\text{trend}}$ quantifies the strength of temporal trends using linear regression.

During the initial bootstrap phase with insufficient data (fewer than 5 samples), the algorithm uses a simplified threshold computation to avoid instability. The trend factor captures whether the metric values are systematically increasing or decreasing, allowing the threshold to adapt accordingly. This prevents false negatives during periods of natural system evolution while maintaining sensitivity to abrupt changes.

The adaptive nature of this threshold computation is crucial for handling diverse auditory environments with varying baseline activity levels, ensuring consistent detection performance across different acoustic contexts without requiring environment-specific calibration.

\section{Proof for Theorem~\ref{thm:phase_response}}
\label{appendix:phase_response_proof}

\subsection{Eigenvalue Analysis and Local Resonance Frequency}

The OWM system is governed by the discrete formulation in Eqs.~\ref{eq:disc_wave_eq}. For a local grid point $(x,y)$, define the state vector $\mathbf{h}=[p,v_x,v_y]^T$. The update can be written as
\begin{equation}
\mathbf{h}(t+\Delta t)=\mathbf{A}\mathbf{h}(t)+S(x,y,t),
\label{eq:discrete_state}
\end{equation}
where
\begin{align}
\mathbf{A}&=\mathbf{M}_{\text{Damp}}^{-1}\mathbf{M}_{\text{Velocity}}, \label{eq:system_matrix}\\
\mathbf{M}_{\text{Velocity}}&=\begin{bmatrix}
1 & -c^2\Delta t\frac{\partial}{\partial x} & -c^2\Delta t\frac{\partial}{\partial y}\\
-\Delta t\frac{\partial}{\partial x} & 1 & 0\\
-\Delta t\frac{\partial}{\partial y} & 0 & 1
\end{bmatrix}, \label{eq:M_velocity}\\
\mathbf{M}_{\text{Damp}}&=\begin{bmatrix}
1+\Delta t\,k^p(x,y) & 0 & 0\\
0 & 1+\Delta t\,k^v(x,y) & 0\\
0 & 0 & 1+\Delta t\,k^v(x,y)
\end{bmatrix}. \label{eq:M_damp}
\end{align}

In Fourier space, $\partial/\partial x\rightarrow i\xi_x$ and $\partial/\partial y\rightarrow i\xi_y$. The corresponding system matrix becomes
\begin{equation}
\mathbf{A}=\mathbf{M}_{\text{Damp}}^{-1}\mathbf{M}_{\text{Velocity}}
=\begin{bmatrix}
\frac{1}{1+\Delta t k^p_{i,j}} & -\frac{c^2_{i,j}\Delta t i\xi_x}{1+\Delta t k^p_{i,j}} & -\frac{c^2_{i,j}\Delta t i\xi_y}{1+\Delta t k^p_{i,j}}\\
-\frac{\Delta t i\xi_x}{1+\Delta t k^v_{i,j}} & \frac{1}{1+\Delta t k^v_{i,j}} & 0\\
-\frac{\Delta t i\xi_y}{1+\Delta t k^v_{i,j}} & 0 & \frac{1}{1+\Delta t k^v_{i,j}}
\end{bmatrix}.
\label{eq:A_fourier}
\end{equation}

The characteristic equation $\det(\mathbf{A}-\lambda\mathbf{I})=0$ expands to
\begin{equation}
\left(\frac{1}{1+\Delta t k^v_{i,j}}-\lambda\right)^2
\left[
\left(\frac{1}{1+\Delta t k^p_{i,j}}-\lambda\right)\left(\frac{1}{1+\Delta t k^v_{i,j}}-\lambda\right)
+
\frac{c^2_{i,j}\Delta t^2(\xi_x^2+\xi_y^2)}{(1+\Delta t k^p_{i,j})(1+\Delta t k^v_{i,j})}
\right]=0.
\label{eq:characteristic}
\end{equation}
This yields one real eigenvalue $\lambda_1=1/(1+\Delta t k^v_{i,j})$ and two eigenvalues satisfying
\begin{equation}
\left(\frac{1}{1+\Delta t k^p_{i,j}}-\lambda\right)\left(\frac{1}{1+\Delta t k^v_{i,j}}-\lambda\right)
+
\frac{c^2_{i,j}\Delta t^2(\xi_x^2+\xi_y^2)}{(1+\Delta t k^p_{i,j})(1+\Delta t k^v_{i,j})}=0.
\label{eq:quadratic}
\end{equation}
For the parameter regime used ($\Delta t k_p = \Delta t k_v = 0.1 \ll 1$ and typical spatial wavelengths), the discriminant of Eq.~\ref{eq:quadratic} satisfies $\Delta < 0$, ensuring complex conjugate eigenvalues. Under the balanced damping condition $k_p = k_v$, the real parts cancel exactly, yielding: 
\begin{equation}
\lambda_{2,3}=
\frac{1}{2}\left[\frac{1}{1+\Delta t k^p_{i,j}}+\frac{1}{1+\Delta t k^v_{i,j}}\right]
\pm
i\frac{c_{i,j}\Delta t\sqrt{\xi_x^2+\xi_y^2}}{\sqrt{(1+\Delta t k^p_{i,j})(1+\Delta t k^v_{i,j})}}.
\label{eq:complex_eigenvalues}
\end{equation}
Let $\theta$ be the phase angle of $\lambda_{2,3}$:
\begin{equation}
\theta=\tan^{-1}\!\left(\frac{\mathrm{Im}(\lambda)}{\mathrm{Re}(\lambda)}\right)
=
\tan^{-1}\!\left(
\frac{2c_{i,j}\Delta t\sqrt{\xi_x^2+\xi_y^2}}
{\sqrt{(1+\Delta t k^p_{i,j})(1+\Delta t k^v_{i,j})}
\left(\frac{1}{1+\Delta t k^p_{i,j}}+\frac{1}{1+\Delta t k^v_{i,j}}\right)}
\right).
\label{eq:phase_angle}
\end{equation}
In discrete time, the oscillation frequency associated with phase $\theta$ is $f=\theta/(\pi\Delta t)$, hence the local characteristic frequency is
\begin{equation}
f_{i,j}=
\frac{1}{\pi\Delta t}
\tan^{-1}\!\left(
\frac{2c_{i,j}\Delta t\sqrt{\xi_x^2+\xi_y^2}}
{\sqrt{(1+\Delta t k^p_{i,j})(1+\Delta t k^v_{i,j})}
\left(\frac{1}{1+\Delta t k^p_{i,j}}+\frac{1}{1+\Delta t k^v_{i,j}}\right)}
\right).
\label{eq:eigenfrequency}
\end{equation}
We define the local resonance frequency of the discrete system by
\(
\omega_{\mathrm{res}}(x,y)=2\pi f_{i,j}.
\label{eq:local_res_freq}
\)

\subsection{Frequency-Domain Response and Phase Delay}

Consider harmonic excitation at angular frequency $\omega$. In the frequency domain, the continuous system in Eqs.~\eqref{eq:wave_eq} satisfies
\begin{align}
[i\omega+k^p(x,y)]p(x,y,\omega)&=-c^2(x,y)\nabla\cdot\mathbf{v}(x,y,\omega)+S_0(x,y), \label{eq:freq_p_continuous}\\
[i\omega+k^v(x,y)]\mathbf{v}(x,y,\omega)&=-\nabla p(x,y,\omega). \label{eq:freq_v_continuous}
\end{align}
Eliminating $\mathbf{v}$ gives
\(
[i\omega+k^p][i\omega+k^v]p-c^2\nabla^2 p=[i\omega+k^v]S_0.
\)

The resonance frequency $\omega_{\mathrm{res}}(x,y)$ is the natural angular frequency of the discrete system from the previous subsection. Using $\omega_{\mathrm{res}}=\theta/\Delta t$ and the small $\Delta t$ approximation yields
\begin{equation}
\omega_{\mathrm{res}}(x,y)\approx c(x,y)\sqrt{\xi_x^2+\xi_y^2}.
\label{eq:w_res_def}
\end{equation}

For slowly varying parameters, we use the local approximation $-c^2\nabla^2 \rightarrow \omega_{\mathrm{res}}^2(x,y)$, leading to the localized frequency-domain relation
\begin{equation}
\Big[(\omega_{\mathrm{res}}^2-k^pk^v-\omega^2)+i\omega(k^p+k^v)\Big]\,p \approx (i\omega+k^v)S_0.
\label{eq:local_freq_relation}
\end{equation}
Define the effective local resonance by
\(
\Omega_{\mathrm{res}}^2(x,y)\triangleq \omega_{\mathrm{res}}^2(x,y)-k^p(x,y)k^v(x,y).
\)
Then the complex denominator governing the response is
\begin{equation}
\mathcal{D}(\omega)=(\Omega_{\mathrm{res}}^2-\omega^2)+i\omega(k^p+k^v).
\label{eq:complex_denominator}
\end{equation}
Hence the amplitude response satisfies
\(
A(x,y,\omega)=|i\omega+k^v|/|\mathcal{D}(\omega)|.
\)

The phase of $p$ is given by the argument of $\mathcal{D}(\omega)$. Writing
\(
\phi(\omega;x,y)=-\arg\mathcal{D}(\omega)
\),
one obtains
\begin{equation}
\phi(\omega;x,y)
=
-\arctan\!\left(
\frac{\omega(k^p+k^v)}{\Omega_{\mathrm{res}}^2-\omega^2}
\right).
\label{eq:phase_from_denominator}
\end{equation}
Near resonance, linearize $\Omega_{\mathrm{res}}^2-\omega^2 \approx -2\Omega_{\mathrm{res}}(\omega-\Omega_{\mathrm{res}})$ to obtain
\begin{equation}
\phi(\omega;x,y)
\approx
\arctan\!\left(
\frac{\omega(k^p+k^v)}{2\Omega_{\mathrm{res}}(\omega-\Omega_{\mathrm{res}})}
\right).
\label{eq:phase_near_res}
\end{equation}
Up to a slowly varying prefactor (absorbed by the proportionality in Eq.~\eqref{eq:phase_delay}), this yields the temporal phase delay. Using the identity $\arctan(-x) = -\arctan(x)$ and absorbing the sign into the definition of positive phase delay (where lag is defined as positive), we obtain:
\begin{equation}
\tau(\omega;x,y)=\frac{\phi(\omega;x,y)}{\omega}
\propto
\arctan\!\left(
\frac{\omega-\omega_{\mathrm{res}}(x,y)}{k^p(x,y)+k^v(x,y)}
\right),
\label{eq:time_delay_derivation}
\end{equation}
which matches the form in Eq.~\eqref{eq:phase_delay}.

Finally, by linearity and superposition, for multi-frequency excitation $S(x,y,t)=\sum_i S_i(x,y,t)$ with narrowband components at $\omega_i$ and under the quasi-steady assumption (envelope variation slower than $(k^p+k^v)^{-1}$), each component is locally delayed, giving
\begin{equation}
p(x,y,t)\approx \sum_i A_i(x,y)\,S_i\big(x,y,t-\tau_i(x,y)\big),
\label{eq:phase_response_appendix}
\end{equation}
which establishes Eq.~\eqref{eq:phase_response}.

\subsection{Application to OWM and Generality}

In OWM (Section~\ref{subsection:experimental_setting} and Appendix~\ref{app:freq_assignment}), $c(x,y)$ is designed to vary spatially, which makes $\omega_{\mathrm{res}}(x,y)=2\pi f_{i,j}$ spatially varying through Eq.~\eqref{eq:eigenfrequency}. For a component at frequency $\omega_i$, Eq.~\eqref{eq:time_delay_derivation} gives a location-dependent delay $\tau_i(x,y)$, and superposition yields the delayed-sum form in Eq.~\eqref{eq:phase_response}.

The proof relies on (i) linear operation so that superposition holds, (ii) the quasi-steady condition on the excitation envelope relative to $(k^p+k^v)^{-1}$, and (iii) slowly varying spatial parameters so that the local approximation in Eq.~\eqref{eq:local_freq_relation} is valid. Under these conditions, different input frequencies induce different local phase delays determined by the mismatch between $\omega_i$ and $\omega_{\mathrm{res}}(x,y)$.

\section{Proof for Theorem~\ref{thm:second_order_equivalence}}
\label{appendix:Second-Order Equivalence}

\begin{proof}
We prove the equivalence by transforming the system step by step. 
We begin by differentiating the pressure equation from Eqs.~\ref{eq:wave_eq} 
with respect to time:

\begin{equation}
\frac{\partial^2 p}{\partial t^2} + k^{p} \frac{\partial p}{\partial t} 
= -\frac{\partial}{\partial t}\left[\nabla \cdot (c^2(x,y) \mathbf{v})\right] 
+ \frac{\partial S}{\partial t}.
\label{eq:E1}
\end{equation}

Since $c^2(x,y)$ is time-independent, we apply the product rule for 
the divergence of a time-dependent vector field with spatially-varying 
coefficients:

\begin{equation}
\frac{\partial}{\partial t}[\nabla \cdot (c^2 \mathbf{v})] 
= \nabla \cdot \left(c^2 \frac{\partial \mathbf{v}}{\partial t}\right) 
+ \frac{\partial \mathbf{v}}{\partial t} \cdot \nabla c^2.
\label{eq:E2}
\end{equation}

To evaluate the velocity derivatives, we use the velocity equations 
from Eqs.~\ref{eq:wave_eq}:

\begin{equation}
\frac{\partial \mathbf{v}}{\partial t} = -k^{v} \mathbf{v} - \nabla p.
\label{eq:E3}
\end{equation}

Substituting Eq.~\ref{eq:E3} into the divergence term in Eq.~\ref{eq:E2}:

\begin{align}
\nabla \cdot \left(c^2 \frac{\partial \mathbf{v}}{\partial t}\right)
&= \nabla \cdot \left[c^2(-k^{v} \mathbf{v} - \nabla p)\right] \nonumber\\
&= -k^{v} \nabla \cdot (c^2 \mathbf{v}) - \nabla \cdot (c^2 \nabla p).
\label{eq:E4}
\end{align}

For the gradient coupling term in Eq.~\ref{eq:E2}:

\begin{equation}
\frac{\partial \mathbf{v}}{\partial t} \cdot \nabla c^2 
= (-k^{v} \mathbf{v} - \nabla p) \cdot \nabla c^2 
= -k^{v} \mathbf{v} \cdot \nabla c^2 - \nabla p \cdot \nabla c^2.
\label{eq:E5}
\end{equation}

Combining Eqs.~\ref{eq:E4} and \ref{eq:E5} into Eq.~\ref{eq:E2}:

\begin{equation}
\frac{\partial}{\partial t}[\nabla \cdot (c^2 \mathbf{v})] 
= -k^{v} \nabla \cdot (c^2 \mathbf{v}) - \nabla \cdot (c^2 \nabla p) 
- k^{v} \mathbf{v} \cdot \nabla c^2 - \nabla p \cdot \nabla c^2.
\label{eq:E6}
\end{equation}

We now eliminate the velocity divergence using the pressure equation 
from Eqs.~\ref{eq:wave_eq}:

\begin{equation}
\nabla \cdot (c^2 \mathbf{v}) = -\frac{\partial p}{\partial t} - k^{p} p + S.
\label{eq:E7}
\end{equation}

Inserting Eq.~\ref{eq:E7} into Eq.~\ref{eq:E6}:

\begin{align}
\frac{\partial}{\partial t}[\nabla \cdot (c^2 \mathbf{v})] 
&= -k^{v}\left(-\frac{\partial p}{\partial t} - k^{p} p + S\right) 
- \nabla \cdot (c^2 \nabla p) \nonumber\\
&\quad - k^{v} \mathbf{v} \cdot \nabla c^2 - \nabla p \cdot \nabla c^2 \nonumber\\
&= k^{v} \frac{\partial p}{\partial t} + k^{v} k^{p} p - k^{v} S 
- \nabla \cdot (c^2 \nabla p) \nonumber\\
&\quad - k^{v} \mathbf{v} \cdot \nabla c^2 - \nabla p \cdot \nabla c^2.
\label{eq:E8}
\end{align}

Substituting Eq.~\ref{eq:E8} back into Eq.~\ref{eq:E1}:

\begin{align}
\frac{\partial^2 p}{\partial t^2} + k^{p} \frac{\partial p}{\partial t}
&= -\left(k^{v} \frac{\partial p}{\partial t} + k^{v} k^{p} p - k^{v} S 
- \nabla \cdot (c^2 \nabla p) \right. \nonumber\\
&\quad \left. - k^{v} \mathbf{v} \cdot \nabla c^2 - \nabla p \cdot \nabla c^2\right) 
+ \frac{\partial S}{\partial t}.
\label{eq:E9}
\end{align}

Finally, rearranging terms leads to:

\begin{equation}
\frac{\partial^2 p}{\partial t^2} + (k^{p} + k^{v}) \frac{\partial p}{\partial t} 
+ k^{v} k^{p} p = \nabla \cdot (c^2 \nabla p) 
+ (k^{v} \mathbf{v} + \nabla p) \cdot \nabla c^2 
+ \left(k^{v} + \frac{\partial}{\partial t}\right) S,
\label{eq:E10}
\end{equation}

which matches Eq.~\ref{eq:second_order} with effective damping coefficient 
$\gamma = k^{p} + k^{v}$, restoring force coefficient $\mu = k^{v} k^{p}$, 
and spatial coupling term $(k^{v} \mathbf{v} + \nabla p) \cdot \nabla c^2$ 
arising from the spatially-varying wave speed.
\end{proof}

\section{Proof for Theorem~\ref{the:Energy_Evolution_Equation}}
\label{appendix:Energy Evolution Equation}

\begin{proof}
We begin by computing the time derivative of energy:
\begin{equation}
\frac{dE}{dt} = \iint \Big[p \frac{\partial p}{\partial t} + v_x \frac{\partial v_x}{\partial t} + v_y \frac{\partial v_y}{\partial t}\Big] dx\,dy.
\label{eq:energy_derivative}
\end{equation}

Substituting the system equations from Eqs.~\ref{eq:wave_eq} into Eq.~\ref{eq:energy_derivative}, we obtain:
\begin{equation}
\frac{dE}{dt} = \iint \Big[p\big(-k^{p}p - c^2(\tfrac{\partial v_x}{\partial x} + \tfrac{\partial v_y}{\partial y}) + S\big) 
+ v_x\big(-k^{v} v_x - \tfrac{\partial p}{\partial x}\big) 
+ v_y\big(-k^{v} v_y - \tfrac{\partial p}{\partial y}\big)\Big] dx\,dy.
\label{eq:substituted_equations}
\end{equation}

Expanding the terms in Eq.~\ref{eq:substituted_equations} yields:
\begin{equation}
\frac{dE}{dt} = \iint \Big[-k^{p} p^2 - c^2 p(\tfrac{\partial v_x}{\partial x} + \tfrac{\partial v_y}{\partial y}) + pS 
- k^{v} v_x^2 - v_x \tfrac{\partial p}{\partial x} - k^{v} v_y^2 - v_y \tfrac{\partial p}{\partial y}\Big] dx\,dy.
\label{eq:expanded_terms}
\end{equation}

We can group the terms in Eq.~\ref{eq:expanded_terms} to separate dissipation, coupling, and source contributions:
\begin{equation}
\small
\frac{dE}{dt} = \iint \Big[-k^{p} p^2 - k^{v}(v_x^2 + v_y^2)\Big] dx\,dy
+ \iint \Big[-c^2 p(\tfrac{\partial v_x}{\partial x} + \tfrac{\partial v_y}{\partial y}) - v_x \tfrac{\partial p}{\partial x} - v_y \tfrac{\partial p}{\partial y}\Big] dx\,dy
+ \iint pS\, dx\,dy.
\label{eq:grouped_terms}
\end{equation}

To handle the coupling terms in Eq.~\ref{eq:grouped_terms}, we apply integration by parts. For example:
\begin{equation}
\iint v_x \frac{\partial p}{\partial x} dx\,dy 
= v_x p \Big|_{x\ \text{boundaries}} - \iint p \frac{\partial v_x}{\partial x} dx\,dy.
\label{eq:integration_by_parts_x}
\end{equation}
Since periodic boundary conditions imply the boundary term vanishes, we have:
\begin{equation}
\iint v_x \frac{\partial p}{\partial x} dx\,dy = -\iint p \frac{\partial v_x}{\partial x} dx\,dy.
\label{eq:x_coupling_simplified}
\end{equation}
Similarly, for the $y$-component:
\begin{equation}
\iint v_y \frac{\partial p}{\partial y} dx\,dy = -\iint p \frac{\partial v_y}{\partial y} dx\,dy.
\label{eq:y_coupling_simplified}
\end{equation}

Using Eqs.~\ref{eq:x_coupling_simplified} and \ref{eq:y_coupling_simplified}, the coupling terms in Eq.~\ref{eq:grouped_terms} reduce to:
\begin{equation}
\begin{aligned}
&-c^2 p \left(\frac{\partial v_x}{\partial x} + \frac{\partial v_y}{\partial y}\right) 
   - v_x \frac{\partial p}{\partial x} 
   - v_y \frac{\partial p}{\partial y} \\[6pt]
&= -c^2 p \left(\frac{\partial v_x}{\partial x} + \frac{\partial v_y}{\partial y}\right) 
   + p \left(\frac{\partial v_x}{\partial x} + \frac{\partial v_y}{\partial y}\right) \\[6pt]
&= -(c^2 - 1)\, p \left(\frac{\partial v_x}{\partial x} + \frac{\partial v_y}{\partial y}\right).
\end{aligned}
\label{eq:coupling_combined}
\end{equation}

Combining all results, we obtain the final expression:
\begin{equation}
\small
\frac{dE}{dt} = -\iint \left[k^{p} p^2 + k^{v}(v_x^2 + v_y^2)\right] dx\,dy 
- \iint (c^2 - 1)p\left(\frac{\partial v_x}{\partial x} + \frac{\partial v_y}{\partial y}\right) dx\,dy 
+ \iint p S\, dx\,dy,
\label{eq:final_energy_evolution}
\end{equation}
which establishes the claim.
\end{proof}

\section{Proofs for wave propagation speed Optimization}
\label{appendix:wave-speed-opt}

\subsection{Interface Reflection and Contrast Effects}
\label{appendix:contrast}

\begin{lemma}[Interface Reflection Coefficient]
\label{lem:interface-reflection}
Consider a plane wave incident on an interface between two constant-speed media with speeds $c_1$ and $c_2$. With uniform density $\rho=1$ so that impedance $Z=\rho c=c$, the pressure reflection coefficient is
\begin{equation}
R=\frac{c_2-c_1}{c_1+c_2}.
\label{eq:app_interface_R}
\end{equation}
The energy reflection ratio is $E_{\mathrm{reflected}}/E_{\mathrm{incident}}=|R|^2$. Under a fixed speed-sum constraint $c_1+c_2=C$, $|R|^2$ is maximized by extreme (binary) contrasts $(c_1,c_2)\in\{(0,C),(C,0)\}$, while $c_1=c_2=C/2$ yields $R=0$.
\end{lemma}

\begin{proof}
With $\rho=1$, the medium impedances are $Z_i=c_i$, $i=1,2$. Continuity of pressure and normal velocity (using $v=p/Z$) at the interface yields
\begin{equation}
p_1+p_1^{r}=p_2^{t},\qquad \frac{p_1}{Z_1}-\frac{p_1^{r}}{Z_1}=\frac{p_2^{t}}{Z_2},
\label{eq:app_boundary_conditions}
\end{equation}
where subscripts denote incident (no subscript), reflected ($r$), and transmitted ($t$) components.  
Solving these equations gives
\begin{equation}
p_2^{t}=\frac{2Z_2}{Z_1+Z_2}\,p_1,\qquad
\frac{p_1^{r}}{p_1}=\frac{Z_2-Z_1}{Z_1+Z_2}=\frac{c_2-c_1}{c_1+c_2}=R.
\label{eq:app_pt_and_R}
\end{equation}

The reflected-to-incident energy ratio is therefore
\begin{equation}
\frac{E_{\mathrm{reflected}}}{E_{\mathrm{incident}}}=|R|^2=\left(\frac{c_2-c_1}{c_1+c_2}\right)^2.
\label{eq:app_energy_ratio}
\end{equation}
If $c_1+c_2 = C$, maximizing $|R|^2 = \left(\frac{c_2 - c_1}{C}\right)^2$ is equivalent to maximizing $|c_2 - c_1|$. The maximum occurs at the binary extremes $(c_1,c_2) = (0,C)$ or $(C,0)$, for which $|R|\to 1$. Any interior choice yields a smaller $R$. For example, $c_1 = c_2 = C/2$ gives $R = 0$, while $c_1 = 0.9C,\ c_2 = 0.1C$ gives $|R| = 0.8 < 1$.
\end{proof}

\subsection{Modal Coupling Framework}
\label{appendix:modal-coupling}

\begin{lemma}[Modal Coupling Derivation and Selection Rule]
\label{lem:modal-coupling}
In the periodic domain $[0,L_x]\times[0,L_y]$ with orthonormal Fourier basis
\begin{equation}
\phi_{m,n}(x,y)=\frac{1}{\sqrt{A}}\exp\!\left(i\frac{2\pi m x}{L_x}+i\frac{2\pi n y}{L_y}\right),\quad 
A=L_xL_y,
\label{eq:app_basis}
\end{equation}
and wavenumbers
\begin{equation}
k_{m,n}^2=\left(\frac{2\pi m}{L_x}\right)^2+\left(\frac{2\pi n}{L_y}\right)^2,
\label{eq:app_k}
\end{equation}
write $p(x,y,t)=\sum_{m,n} a_{m,n}(t)\phi_{m,n}(x,y)$ and $c^2(x,y)=c_0^2+\delta c^2(x,y)$. Here, $a_{m,n}(t)$ are the time-dependent modal coefficients. Neglecting damping, the modal system
\begin{equation}
\ddot{a}_{m,n}+\omega_{m,n}^2 a_{m,n}=\sum_{m',n'} C_{(m,n),(m',n')} a_{m',n'}+S_{m,n},
\label{eq:app_modal}
\end{equation}
has $\omega_{m,n}^2=c_0^2 k_{m,n}^2$ and coupling
\begin{equation}
C_{(m,n),(m',n')}=k_{m',n'}^2 \,\langle \delta c^2,\phi_{m',n'},\phi_{m,n}\rangle.
\label{eq:app_Cdef}
\end{equation}
For striped $\delta c^2(x,y)=f(y)$,
\begin{equation}
\langle \delta c^2,\phi_{m',n'},\phi_{m,n}\rangle=\delta_{m,m'}\,\hat V_{n-n'},
\label{eq:app_selection}
\end{equation}
i.e., coupling is block-diagonal in $m$ and depends on the $y$-Fourier coefficients $\hat V_q$ of $f$, where
\begin{equation}
\hat V_q=\frac{1}{L_y}\int_0^{L_y} f(y)\,e^{i\frac{2\pi q y}{L_y}}\,dy.
\label{eq:app_Vhat}
\end{equation}
\end{lemma}

\begin{proof}
We begin by establishing the modal expansion and deriving the coupling matrix through systematic projection onto Fourier basis functions.

On the periodic domain $[0,L_x]\times[0,L_y]$, we employ the orthonormal Fourier basis defined in Eq.~\eqref{eq:app_basis}, where $A=L_xL_y$ is the domain area. These functions satisfy $\nabla^2\phi_{m,n} = -k_{m,n}^2 \phi_{m,n}$.

Starting from the wave equation without damping:
\begin{equation}
\frac{\partial^2 p}{\partial t^2} = c^2(x,y)\nabla^2 p + S(x,y,t), \quad 
c^2(x,y) = c_0^2 + \delta c^2(x,y),
\label{eq:app_wave_equation}
\end{equation}
we expand $p(x,y,t)=\sum_{m,n} a_{m,n}(t)\phi_{m,n}(x,y)$ to obtain:
\begin{align}
\sum_{m,n} \ddot{a}_{m,n}(t) \phi_{m,n}(x,y) &= (c_0^2 + \delta c^2(x,y)) \nabla^2 \left(\sum_{m,n} a_{m,n}(t) \phi_{m,n}(x,y)\right) + S(x,y,t) \nonumber \\
&= \sum_{m,n} a_{m,n}(t) (c_0^2 + \delta c^2(x,y)) \nabla^2 \phi_{m,n}(x,y) + S(x,y,t).
\label{eq:app_substituted}
\end{align}

Substituting $\nabla^2 \phi_{m,n} = -k_{m,n}^2 \phi_{m,n}$ into Eq.~\eqref{eq:app_substituted} yields:
\begin{equation}
\sum_{m,n} \ddot a_{m,n}(t)\phi_{m,n} = 
-\sum_{m,n} a_{m,n}(t)k_{m,n}^2(c_0^2+\delta c^2)\phi_{m,n} + S(x,y,t).
\label{eq:app_expanded_wave}
\end{equation}

To project this equation onto individual modes, we define the normalized inner product on $[0,L_x]\times[0,L_y]$ by:
\begin{equation}
\langle f,g\rangle := \frac{1}{A}\int_{0}^{L_x}\!\!\int_{0}^{L_y} f(x,y)\,g^*(x,y)\,dx\,dy,
\qquad A=L_xL_y,
\label{eq:app_inner_product}
\end{equation}
so that the Fourier modes are orthonormal: $\langle \phi_{m,n},\phi_{\mu,\nu}\rangle=\delta_{m,\mu}\delta_{n,\nu}$, where $\delta$ denotes the Kronecker delta function. We also use the trilinear shorthand:
\begin{equation}
\langle \delta c^2,\phi_{m,n},\phi_{\mu,\nu}\rangle 
:= \frac{1}{A}\int_{0}^{L_x}\!\!\int_{0}^{L_y} \delta c^2(x,y)\,\phi_{m,n}^*(x,y)\,\phi_{\mu,\nu}(x,y)\,dx\,dy.
\label{eq:app_trilinear_form}
\end{equation}

Multiplying Eq.~\eqref{eq:app_expanded_wave} by $\phi_{\mu,\nu}(x,y)$ and integrating over the domain, we analyze each term separately. For the left-hand side:
\begin{align}
\left\langle \sum_{m,n}\ddot a_{m,n}\phi_{m,n},\,\phi_{\mu,\nu}\right\rangle
&=\sum_{m,n}\ddot a_{m,n}\,\langle \phi_{m,n},\phi_{\mu,\nu}\rangle
=\sum_{m,n}\ddot a_{m,n}\,\delta_{m,\mu}\delta_{n,\nu}
=\ddot a_{\mu,\nu}(t).
\label{eq:app_lhs_projection}
\end{align}

For the reference-speed part:
\begin{align}
-\left\langle \sum_{m,n}a_{m,n}\,k_{m,n}^2\,c_0^2\,\phi_{m,n},\,\phi_{\mu,\nu}\right\rangle
&=-c_0^2\sum_{m,n}a_{m,n}\,k_{m,n}^2\,\langle \phi_{m,n},\phi_{\mu,\nu}\rangle \nonumber \\
&=-c_0^2\,k_{\mu,\nu}^2\,a_{\mu,\nu}(t).
\label{eq:app_reference_projection}
\end{align}

For the perturbation (coupling) part:
\begin{align}
-\left\langle \sum_{m,n}a_{m,n}\,k_{m,n}^2\,\delta c^2\,\phi_{m,n},\,\phi_{\mu,\nu}\right\rangle
&=-\sum_{m,n}a_{m,n}\,k_{m,n}^2\,\langle \delta c^2,\phi_{m,n},\phi_{\mu,\nu}\rangle.
\label{eq:app_coupling_projection}
\end{align}

For the source part:
\begin{equation}
\langle S(x,y,t),\phi_{\mu,\nu}\rangle =: S_{\mu,\nu}(t).
\label{eq:app_source_projection}
\end{equation}

Collecting terms from Eqs.~\eqref{eq:app_lhs_projection}--\eqref{eq:app_source_projection} gives:
\begin{equation}
\ddot a_{\mu,\nu}(t) + c_0^2 k_{\mu,\nu}^2\,a_{\mu,\nu}(t)
= -\sum_{m,n} k_{m,n}^2\,\langle \delta c^2,\phi_{m,n},\phi_{\mu,\nu}\rangle\,a_{m,n}(t) + S_{\mu,\nu}(t).
\label{eq:app_modal_equation_raw}
\end{equation}

Defining $\omega_{m,n}^2:=c_0^2k_{m,n}^2$ and the coupling matrix:
\begin{equation}
C_{(\mu,\nu),(m,n)} := k_{m,n}^2\,\langle \delta c^2,\phi_{m,n},\phi_{\mu,\nu}\rangle,
\label{eq:app_coupling_matrix_def}
\end{equation}
and relabeling $(\mu,\nu)\to(m,n)$, $(m,n)\to(m',n')$ in the sum, we obtain the modal system in Eq.~\eqref{eq:app_modal}.

For the special case of striped distributions where $\delta c^2(x,y)=f(y)$, we can evaluate the coupling matrix explicitly using Eq.~\eqref{eq:app_trilinear_form}:
\begin{align}
\langle \delta c^2,\phi_{m',n'},\phi_{m,n}\rangle 
&= \frac{1}{A}\iint f(y)\,
e^{i\frac{2\pi(m-m')x}{L_x}} e^{i\frac{2\pi(n-n')y}{L_y}}\,dx\,dy \nonumber \\
&= \frac{1}{A}\!\left[\int_0^{L_x} e^{i\frac{2\pi(m-m')x}{L_x}} dx\right]
\left[\int_0^{L_y} f(y) e^{i\frac{2\pi(n-n')y}{L_y}} dy\right].
\label{eq:app_striped_evaluation}
\end{align}

The $x$-integral in Eq.~\eqref{eq:app_striped_evaluation} evaluates to $L_x\delta_{m,m'}$. Defining the Fourier coefficients:
\begin{equation}
\hat V_q = \frac{1}{L_y}\int_0^{L_y} f(y) e^{i\frac{2\pi q y}{L_y}}\,dy,
\quad q=n-n',
\label{eq:app_fourier_coeffs}
\end{equation}
the $y$-integral becomes $L_y\hat V_{n-n'}$. Therefore:
\begin{equation}
\langle \delta c^2,\phi_{m',n'},\phi_{m,n}\rangle 
= \delta_{m,m'}\hat V_{n-n'}.
\label{eq:app_striped_coupling}
\end{equation}

Substituting Eq.~\eqref{eq:app_striped_coupling} into Eq.~\eqref{eq:app_coupling_matrix_def}, the coupling matrix under striped distributions becomes block-diagonal in $m$:
\begin{equation}
C_{(m,n),(m',n')} = k_{m',n'}^2 \delta_{m,m'} \hat V_{n-n'}.
\label{eq:app_block_diagonal_coupling}
\end{equation}
\end{proof}

\subsection{Fourier Coefficient Maximization}
\label{appendix:fourier-max}

\begin{lemma}[Square-Wave Fourier Optimality]
\label{lem:fourier-max}
For a periodic function $f(y)$ with period $L$ and pointwise constraint $|f(y)| \leq A$, the magnitude of the $q$-th Fourier coefficient
\begin{equation}
\hat{f}_q = \frac{1}{L} \int_0^L f(y) e^{i 2\pi q y/L} dy
\label{eq:app_fourier_coeff}
\end{equation}
is maximized by the square wave 
\begin{equation}
f(y) = A \cdot \mathrm{sgn}\left[\cos\left(\frac{2\pi q y}{L}\right)\right],
\label{eq:app_square_wave}
\end{equation}
achieving $|\hat{f}_q| = 2A/(\pi q)$ for odd $q$.
\end{lemma}

\begin{proof}
By the Cauchy-Schwarz inequality,
\begin{equation}
|\hat{f}_q| \leq \frac{1}{L} \int_0^L |f(y)| \, dy \leq \frac{1}{L} \int_0^L A \, dy = A.
\label{eq:app_fq_bound}
\end{equation}

Equality holds when $f(y)$ is phase-aligned with $e^{i 2\pi q y/L}$, i.e., when
\begin{equation}
f(y) = A \cdot \mathrm{sgn}\left[\cos\left(\frac{2\pi q y}{L}\right)\right].
\label{eq:app_fq_optimal}
\end{equation}

For this square wave, direct computation gives
\begin{equation}
\hat{f}_q = \frac{1}{L} \int_0^L A \cdot \mathrm{sgn}\left[\cos\left(\frac{2\pi q y}{L}\right)\right] e^{i 2\pi q y/L} dy = \frac{2A}{\pi q}
\label{eq:app_fq_value}
\end{equation}
for odd $q$, confirming optimality.

Alternatively, this result follows from the convex optimization perspective: the set of functions satisfying $|f(y)| \leq A$ and fixed mean is convex, and the objective $|\hat{f}_q|$ is a linear functional. By the extreme point theorem, the maximum occurs at an extreme point of the feasible set, which corresponds to $f(y) \in \{-A, A\}$ almost everywhere. The phase-alignment condition then determines the switching pattern, yielding the square wave in Eq.~\eqref{eq:app_square_wave}.
\end{proof}

\subsection{Proof of Theorem~\ref{thm:striped-bragg}}
\label{appendix:bragg-proof}

We now prove that the striped square-wave perturbation maximizes drift detection sensitivity through optimal spatial organization.

\begin{proof}
The proof establishes optimality through three complementary perspectives: geometric structure, resonant coupling, and spectral amplitude.

\paragraph{Part I: Geometric constraint from striped structure.}

From Lemma~\ref{lem:modal-coupling}, when the speed perturbation takes the form $\delta c^2(x,y) = f(y)$ (independent of $x$), the coupling matrix satisfies
\begin{equation}
C_{(m,n),(m',n')} = k_{m',n'}^2 \delta_{m,m'} \hat{V}_{n-n'},
\label{eq:proof_coupling}
\end{equation}
where $\hat{V}_q$ are the Fourier coefficients of $f(y)$ defined by
\begin{equation}
\hat{V}_q = \frac{1}{L_y}\int_0^{L_y} f(y)\,e^{i\frac{2\pi q y}{L_y}}\,dy.
\label{eq:proof_Vq}
\end{equation}

The Kronecker delta $\delta_{m,m'}$ in Eq.~\eqref{eq:proof_coupling} reveals that the striped structure decouples modes with different $x$-direction wavenumbers: modes $(m,n)$ and $(m',n')$ with $m\neq m'$ do not couple. Within each $m$-block, coupling between modes $(m,n)$ and $(m,n')$ depends solely on the Fourier coefficient $\hat{V}_{n-n'}$ of the $y$-direction profile $f(y)$.

This geometric decoupling converts the multi-dimensional optimization problem over $c(x,y)$ into independent one-dimensional problems along $y$ for each fixed $m$. This reduction reveals that unidirectional (striped) variation is structurally optimal for maximizing targeted modal coupling.

\paragraph{Part II: Resonant period selection via Fourier harmonics.}

To maximize coupling between modes $(m,n)$ and $(m,n+q_0)$, we examine
\begin{equation}
|C_{(m,n),(m,n+q_0)}| = k_{m,n+q_0}^2 |\hat{V}_{q_0}|.
\label{eq:proof_target}
\end{equation}

Since $k_{m,n+q_0}^2$ is fixed by mode indices, maximizing coupling reduces to maximizing $|\hat{V}_{q_0}|$.

Consider a periodic function $f(y)$ on $[0, L_y]$ with period $T$ (where $T$ divides $L_y$ evenly). The periodicity condition $f(y) = f(y + T)$ implies that the Fourier series Eq.~\eqref{eq:proof_Vq} has $\hat{V}_q = 0$ unless $q$ is an integer multiple of $L_y/T$. Specifically, the Fourier coefficients are nonzero only at
\begin{equation}
q = n \cdot \frac{L_y}{T}, \quad n = 0, \pm 1, \pm 2, \ldots
\label{eq:proof_harmonic_condition}
\end{equation}

For the fundamental Fourier component (largest amplitude in a typical periodic function) to occur at the target mode index $q_0$, we require
\begin{equation}
q_0 = 1 \cdot \frac{L_y}{T} \quad \Rightarrow \quad T = \frac{L_y}{q_0}.
\label{eq:proof_period}
\end{equation}

This period selection has a clear physical interpretation. In the modal expansion, mode index $n$ corresponds to a $y$-direction wavenumber
\begin{equation}
k_y^{(n)} = \frac{2\pi n}{L_y}.
\label{eq:proof_ky_def}
\end{equation}

The discrete modal separation $q_0$ therefore corresponds to a wavenumber difference
\begin{equation}
\Delta k_y = k_y^{(n+q_0)} - k_y^{(n)} = \frac{2\pi (n+q_0)}{L_y} - \frac{2\pi n}{L_y} = \frac{2\pi q_0}{L_y}.
\label{eq:proof_delta_ky}
\end{equation}

Using the wave relation $k = 2\pi/\lambda$, the spatial wavelength associated with this mode coupling is
\begin{equation}
\lambda_{\text{coupling}} = \frac{2\pi}{\Delta k_y} = \frac{2\pi}{2\pi q_0/L_y} = \frac{L_y}{q_0} = T.
\label{eq:proof_wavelength_detailed}
\end{equation}

When the stripe period $T$ matches this coupling wavelength, we achieve resonant enhancement analogous to the Bragg condition for coherent scattering in periodic structures~\cite{kushwaha1993acoustic}: reflections from successive interfaces interfere constructively, maximizing energy transfer between the mode pair $(m,n)$ and $(m,n+q_0)$.

The period $T = L_y/q_0$ is not merely sufficient but necessary for fundamental-harmonic coupling. Any other period would either shift the fundamental component to a different mode index or distribute energy across multiple harmonics, reducing coupling strength at $q_0$.

\paragraph{Part III: Amplitude maximization through square-wave profile.}

Given the resonant period $T = L_y/q_0$ from Part II, we now optimize the profile shape $f(y)$ to maximize $|\hat{V}_{q_0}|$ under constraint $|f(y)| \leq A$.

For a function $f(y)$ with period $T$, we can compute $\hat{V}_{q_0}$ by integrating over one period and accounting for the number of periods in $[0, L_y]$. Since $f(y)$ repeats $q_0$ times over $[0, L_y]$, we have
\begin{equation}
\hat{V}_{q_0} = \frac{1}{L_y} \int_0^{L_y} f(y) e^{i 2\pi q_0 y/L_y} dy = \frac{q_0}{L_y} \int_0^{T} f(y) e^{i 2\pi q_0 y/L_y} dy.
\label{eq:proof_Vq0_integral}
\end{equation}

Substituting $T = L_y/q_0$ and changing variables $\eta = q_0 y/L_y$ (so $dy = T d\eta$ and $\eta \in [0,1]$ over one period):
\begin{equation}
\hat{V}_{q_0} = \frac{q_0}{L_y} \int_0^1 f(T\eta) e^{i 2\pi \eta} T \, d\eta = \frac{1}{L_y/q_0} \int_0^{L_y/q_0} f(y) e^{i 2\pi q_0 y/L_y} dy.
\label{eq:proof_change_var}
\end{equation}

Now we apply Lemma~\ref{lem:fourier-max}. Define a rescaled Fourier coefficient for one period:
\begin{equation}
\tilde{V} := \frac{1}{T} \int_0^T f(y) e^{i 2\pi q_0 y/L_y} dy = \frac{1}{T} \int_0^T f(y) e^{i 2\pi y/T} dy,
\label{eq:proof_rescaled}
\end{equation}
where we used $q_0/L_y = 1/T$. This is precisely the Fourier coefficient at the fundamental frequency for a function with period $T$, matching the form in Lemma~\ref{lem:fourier-max} with $L = T$ and $q = 1$.

From Lemma~\ref{lem:fourier-max}, $|\tilde{V}|$ is maximized when $f(y)$ is the square wave
\begin{equation}
f(y) = A \cdot \mathrm{sgn}\left[\cos\left(\frac{2\pi y}{T}\right)\right] = A \cdot \mathrm{sgn}\left[\cos\left(\frac{2\pi q_0 y}{L_y}\right)\right],
\label{eq:proof_square_wave}
\end{equation}
achieving
\begin{equation}
|\tilde{V}|_{\max} = \frac{2A}{\pi \cdot 1} = \frac{2A}{\pi}.
\label{eq:proof_max_rescaled}
\end{equation}

From Eq.~\eqref{eq:proof_change_var}, we have $\hat{V}_{q_0} = \tilde{V}$, therefore
\begin{equation}
|\hat{V}_{q_0}|_{\max} = \frac{2A}{\pi}.
\label{eq:proof_max_coefficient}
\end{equation}

The optimality of the square-wave profile admits dual interpretations. Mathematically, the Cauchy-Schwarz inequality (used in proving Lemma~\ref{lem:fourier-max}) shows that phase alignment with the Fourier kernel $e^{i 2\pi q_0 y/L_y}$ while maintaining maximal amplitude $|f(y)| = A$ maximizes the overlap integral in Eq.~\eqref{eq:proof_Vq}.

Physically, the square wave creates binary impedance contrasts at interfaces $y = (2n+1)T/2$ for integer $n$. At each interface, waves reflect with coefficient approaching $R \to \pm 1$ as shown in Lemma~\ref{lem:interface-reflection}. When interfaces are spaced at period $T = L_y/q_0$, reflections from all interfaces arrive in phase at the coupled mode pair, creating maximal constructive interference. The Fourier coefficient $|\hat{V}_{q_0}|$ quantifies precisely this coherent amplification.

\paragraph{Synthesis: Joint optimization via amplitude and phase.}

Combining Parts I--III, the striped square wave
\begin{equation}
\delta c^2(x,y) = A \cdot \mathrm{sgn}\left[\cos\left(\frac{2\pi q_0 y}{L_y}\right)\right]
\label{eq:proof_final}
\end{equation}
simultaneously achieves: (\textit{i}) maximal modal coupling amplitude $|C_{(m,n),(m,n+q_0)}| = k_{m,n+q_0}^2 \cdot 2A/\pi$, and (\textit{ii}) maximal spatial differentiation with dynamic range $\Delta c = 2A$. The latter creates maximal spatial variation in local resonance frequencies $\omega_{\mathrm{res}}(y) \propto c(y)$ and consequently in frequency-dependent phase delays $\tau_i(y)$ (Theorem~\ref{thm:phase_response}). Different input frequencies thus produce spatially distinct phase patterns, enabling frequency-selective detection. This joint optimization of coupling amplitude and phase selectivity maximizes drift detection sensitivity.
\end{proof}

\section{Prompt}
\label{appendix:prompt}

The prompt used to make the classification is shown in Fig.~\ref{fig:prompt-violence-audio}. To ensure a rigorous evaluation on the XD-Violence dataset \cite{wu2020not}, we adopt the ground-truth label definitions provided by the original benchmark authors. 

A significant challenge in applying zero-shot Audio Language Models (ALMs) to this dataset is the nature of the source material; the majority of audio scenarios in XD-Violence are extracted from action movies and films. We observed that without domain-specific pre-training or fine-tuning, general-purpose ALMs tend to exhibit a ``fictional bias.'' Specifically, the models often recognize the high-energy acoustic signatures or specific dialogue as being part of a movie or video game and, as a result, classify the sample as containing ``no violence'' because it is perceived as non-real.

To mitigate this, we designed a specific Fictional Context Rule within our prompt. This rule explicitly instructs the model to ignore any mentions of the audio being from a fictional medium and to focus solely on the physical actions described (e.g., classifying a movie gunshot as a valid shooting event). This ensures that the ALM functions as a security-oriented classifier rather than a content-type identifier, allowing for a fair comparison against supervised baselines.

\begin{figure}[H]
\vspace{-0.2em}
\centering
\setlength{\fboxsep}{6pt}
\setlength{\fboxrule}{0.6pt}
\fbox{%
\begin{tabular}{@{}p{0.96\linewidth}@{}}
\small \textbf{XD-Violence Audio Track Classification Prompt}\\[0.4em]
\hrule\\[-0.2em]

\textbf{[ROLE]}\\
You are a professional security audio classifier.\\
Your ONLY task is to detect and categorize violent events from audio descriptions.\\[0.6em]

\textbf{[EVIDENCE]}\\
\texttt{\{evidence\_text\}}\\[0.6em]

\textbf{[RULES]}\\
\textbf{1. Label mapping (mandatory):}\\
\begin{tabular}{@{}ll@{}}
\textbf{B1}: & Fighting (scuffles, hitting, physical combat)\\
\textbf{B2}: & Shooting (gunfire, gunshots)\\
\textbf{B4}: & Riot (large crowd chaos, protests, shouting)\\
\textbf{B5}: & Abuse (harassment, physical or verbal abuse)\\
\textbf{B6}: & Car accident (crashes, breaking glass, vehicle impacts)\\
\textbf{G}:  & Explosion (blasts, loud booming sounds)\\
\textbf{None}: & Only if absolutely no violent events are mentioned.\\
\end{tabular}\\[0.6em]

\textbf{2. Violence dominance rule:}\\
If any single part of the evidence mentions a violent event, you MUST classify the sample based on that event,
even if all other parts of the evidence state ``there is no violence.'' One positive detection overrides all
\textbf{None} detections.\\[0.6em]

\textbf{3. Fictional context rule:}\\
Strictly IGNORE any mentions of the audio being from a ``video game,'' ``movie,'' ``film,'' or ``soundtrack.'' 
You must classify based on the actions described (e.g., if a movie gunshot is described, label it \textbf{B2}).\\[0.6em]

\textbf{[OUTPUT FORMAT]}\\
Output ONLY the code (e.g., \textbf{B4}). Do not include reasoning, descriptions, punctuation, or extra text.\\

\end{tabular}%
}
\vspace{-0.4em}
\caption{Prompt used for violence event classification from audio descriptions.}
\label{fig:prompt-violence-audio}
\end{figure}

\section{Additional Qualitative Comparisons}
\label{appendix:Additional Qualitative Comparisons}

To complement the illustrative cases in Section \ref{subsection:examples}, we provide six additional Mel-spectrogram examples of OWM. These examples highlight diverse acoustic scenarios and further demonstrate the relative robustness of OWM.

\textbf{Car horn onset (Example R0003).} As shown in Fig.~\ref{fig:R0003}, a car horn appears at 21\,s and persists for several seconds. Both methods successfully detect the onset, but OWM localizes the event with fewer spurious triggers.

\begin{figure}[H]
    \centering
    \includegraphics[width=0.7\textwidth]{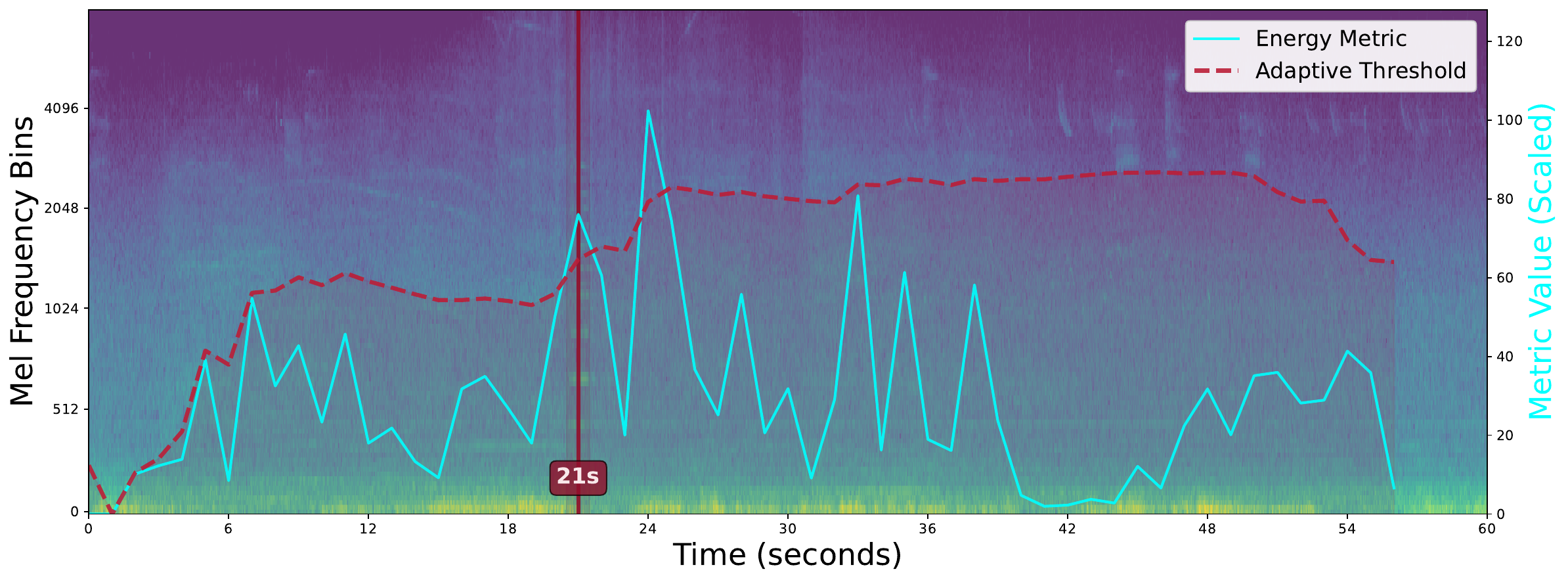}
    \caption{Example R0003: Car horn onset. OWM localizes events with few spurious triggers.}
    \label{fig:R0003}
\end{figure}

\textbf{Intermittent piano playing (Example R0007).} As shown in Fig.~\ref{fig:R0007}, piano notes occur intermittently, producing pauses between onsets. OWM registers four discrete activations, thereby reducing sensitivity to the continuous drift caused by spectral divergence.

\begin{figure}[H]
    \centering
    \includegraphics[width=0.7\textwidth]{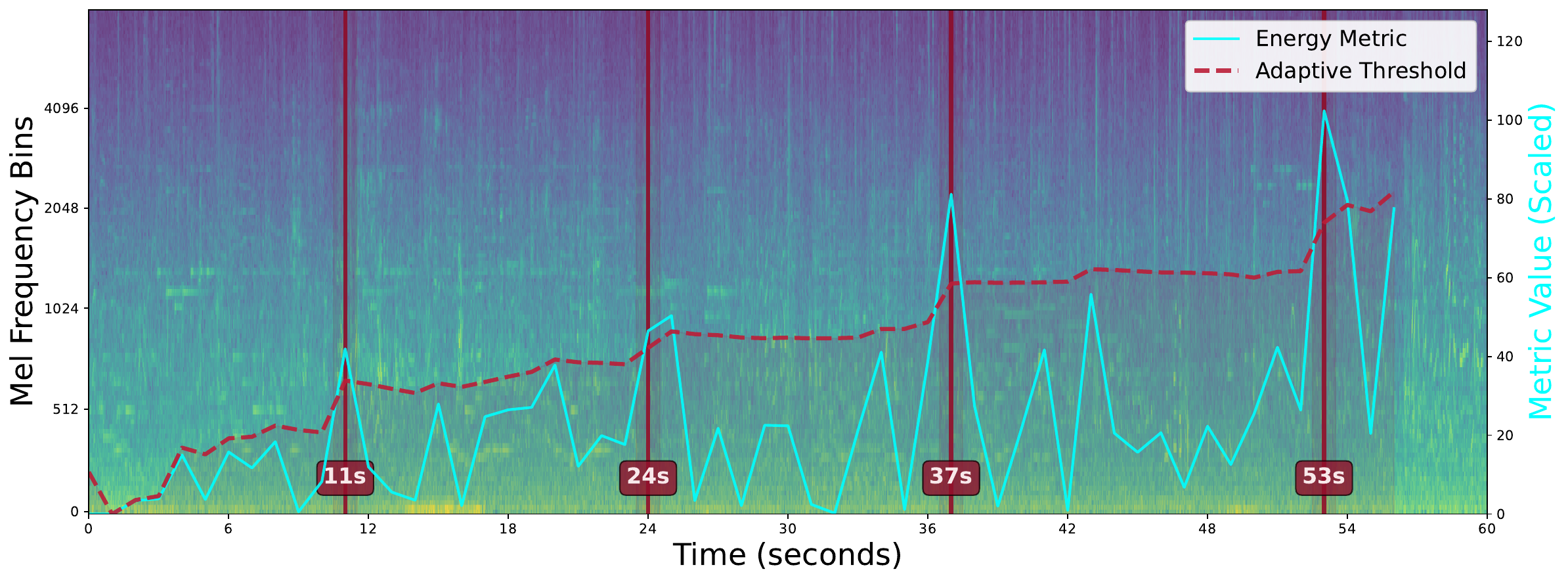}
    \caption{Example R0007 (Chalet du Mont-Royal, Montreal — street traffic):The most salient event is intermittent piano playing, which triggers OWM four times; this discretization reduces sensitivity to the continuous drift caused by spectral divergence.}
    \label{fig:R0007}
\end{figure}

\textbf{Conversational speech (Example R0028).} As shown in Fig.~\ref{fig:R0028}, after 30\,s, a dialogue with pauses and speaker changes begins. OWM detects a single event at 38\,s, reflecting its robustness to intra-speech variability.

\begin{figure}[H]
    \centering
    \includegraphics[width=0.7\textwidth]{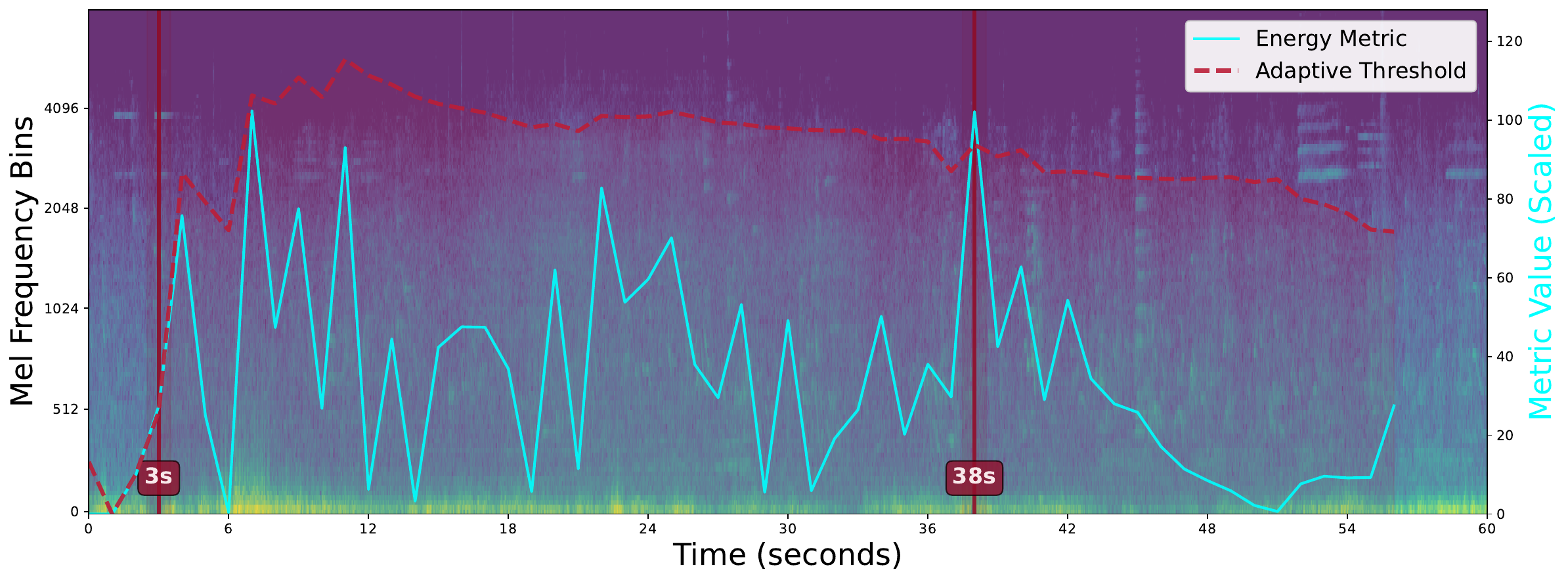}
    \caption{Example R0028 (Heping Road, Tianjin — street traffic): The most salient event is conversation after 30\,s, with pauses and speaker changes. OWM is triggered only once at 38\,s.}
    \label{fig:R0028}
\end{figure}

\textbf{Repeated car horn sounds (Example R0030).} As shown in Fig.~\ref{fig:R0030}, horns occur around 24\,s, 32\,s, and 50\,s. OWM detects only the late instance (around 51\,s).

\begin{figure}[H]
    \centering
    \includegraphics[width=0.7\textwidth]{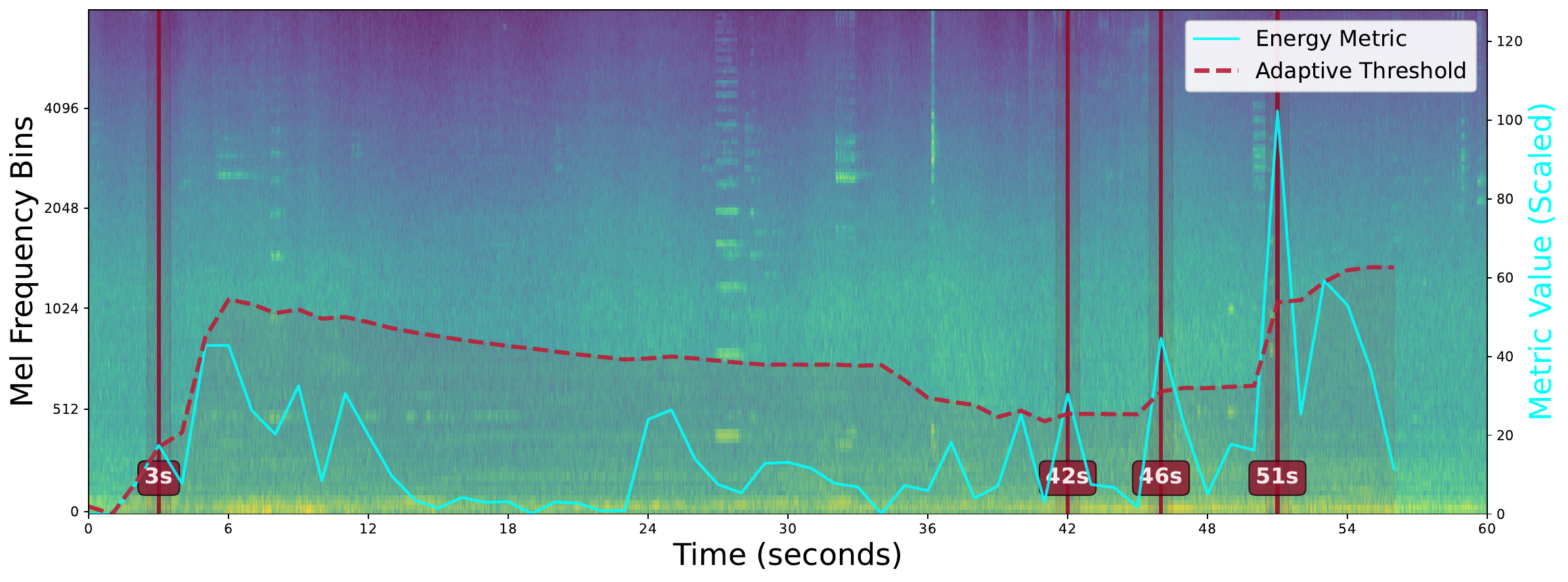}
    \caption{Example R0030 (Century Clock, Tianjin — street traffic): The most salient event is repeated car horn sounds at approximately 24\,s, 32\,s, and 50\,s. OWM detects only the instance around 51\,s.}
    \label{fig:R0030}
\end{figure}

\textbf{Railway station announcement (Example R0031).} As shown in Fig.~\ref{fig:R0031}, the segment contains continuous speech without other salient events. OWM triggers only once during the initial threshold adaptation phase, demonstrating its robustness against continuous background speech.

\begin{figure}[H]
    \centering
    \includegraphics[width=0.7\textwidth]{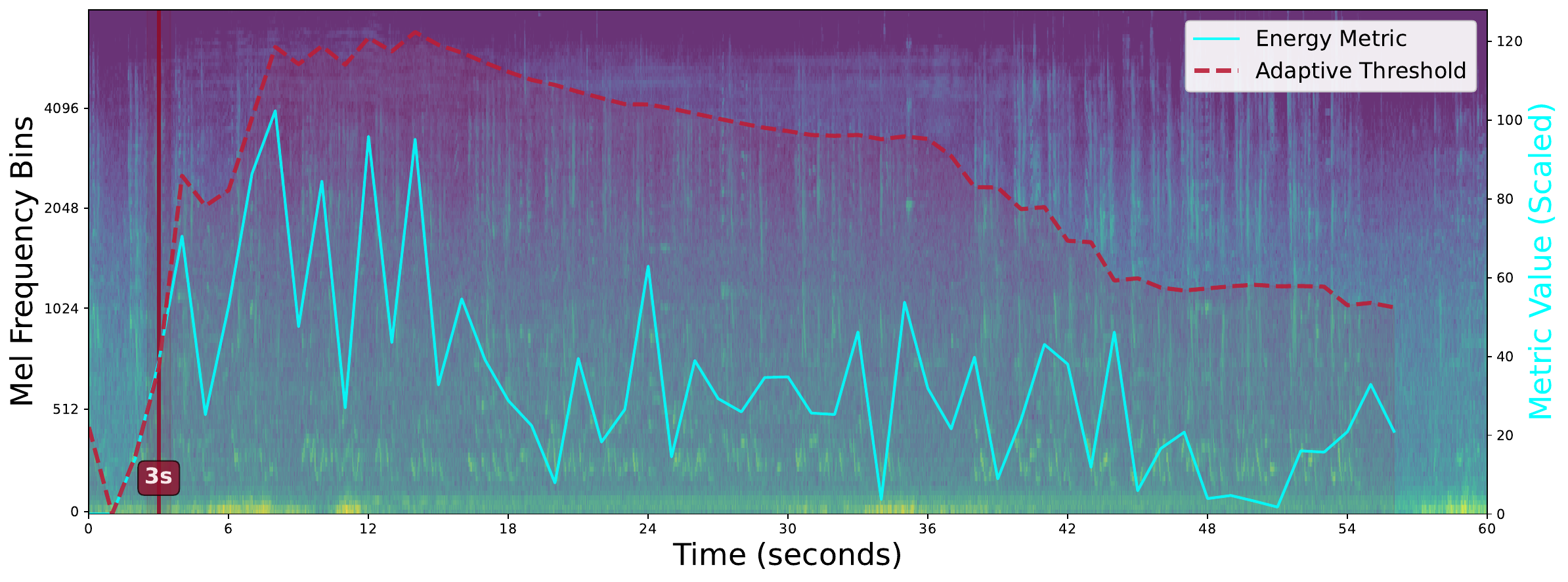}
    \caption{Example R0031 (Tianjin Railway Station, Tianjin — public square): Analysis of a railway station announcement. OWM triggers only once during the initial threshold adaptation, demonstrating its stability and resistance to false activations during continuous speech.}
    \label{fig:R0031}
\end{figure}

\textbf{Church bell with crowd talking (Example R0131).} As shown in Fig.~\ref{fig:R0131}, the spectrogram is dominated by continuous bell ringing. OWM’s adaptive threshold rises in response to this persistent signal, effectively suppressing background crowd noise and preventing spurious detections.

\begin{figure}[H]
    \centering
    \includegraphics[width=0.7\textwidth]{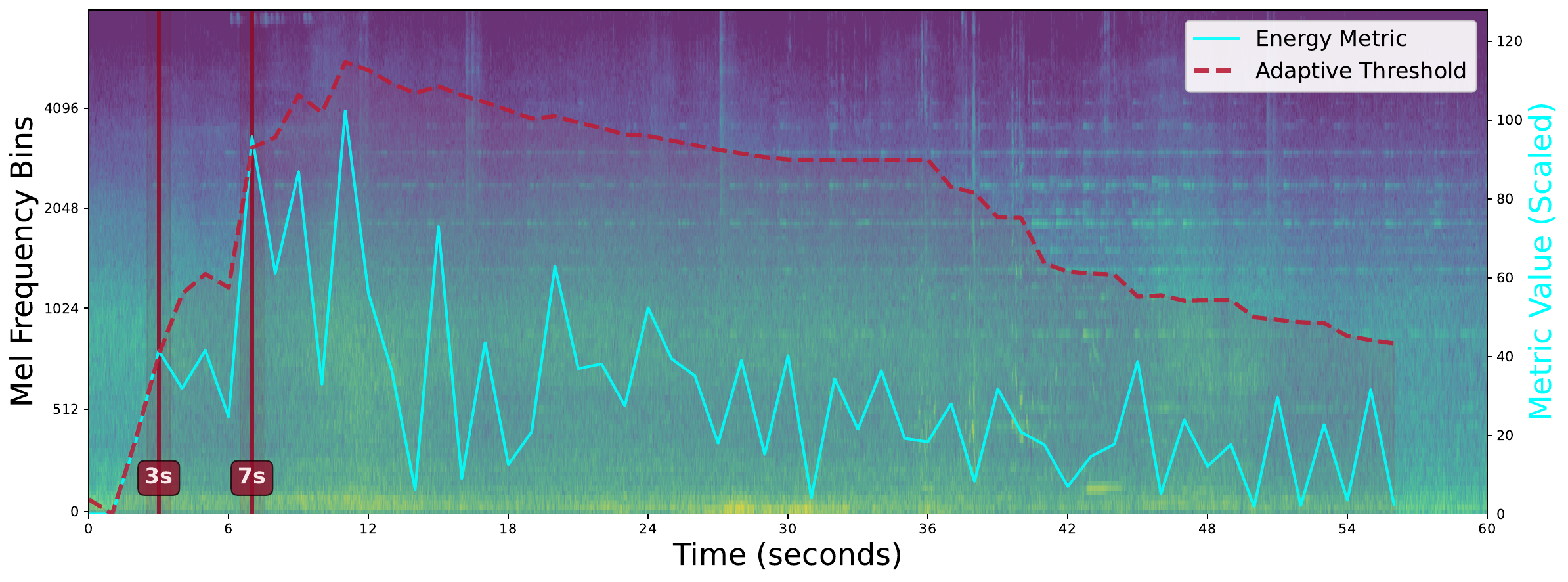}
    \caption{Example R0131 (Town Hall Square, Vilnius — square): OWM response to continuous church bell ringing. The adaptive threshold scales according to the persistent bell signal, effectively suppressing background crowd noise and preventing spurious detections.}
    \label{fig:R0131}
\end{figure}

Taken together, these additional examples reinforce the primary findings: OWM maintains high specificity by avoiding spurious detections during transient pauses, repeated motifs, or continuous sources, while remaining sensitive to salient novel events.

\section{Additional FFT Analysis}
\label{appendix:Extra FFT Analysis}

\begin{figure}[H]
    \centering
    \begin{subfigure}{0.2\textwidth}
        \includegraphics[width=\linewidth,trim=0 0 0 60,clip]{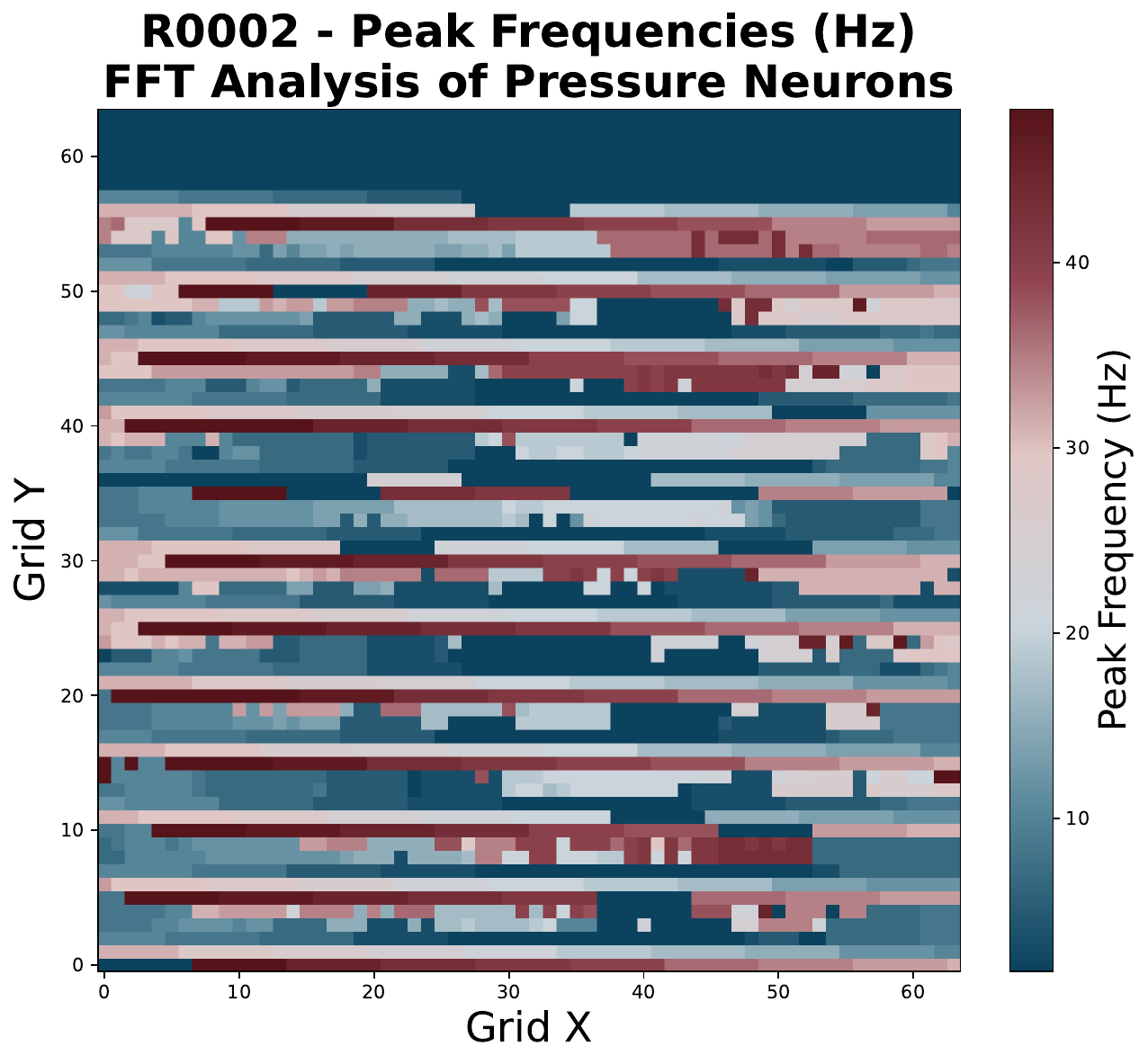}
        \caption{R0002}
    \end{subfigure}\hfill
    \begin{subfigure}{0.2\textwidth}
        \includegraphics[width=\linewidth,trim=0 0 0 60,clip]{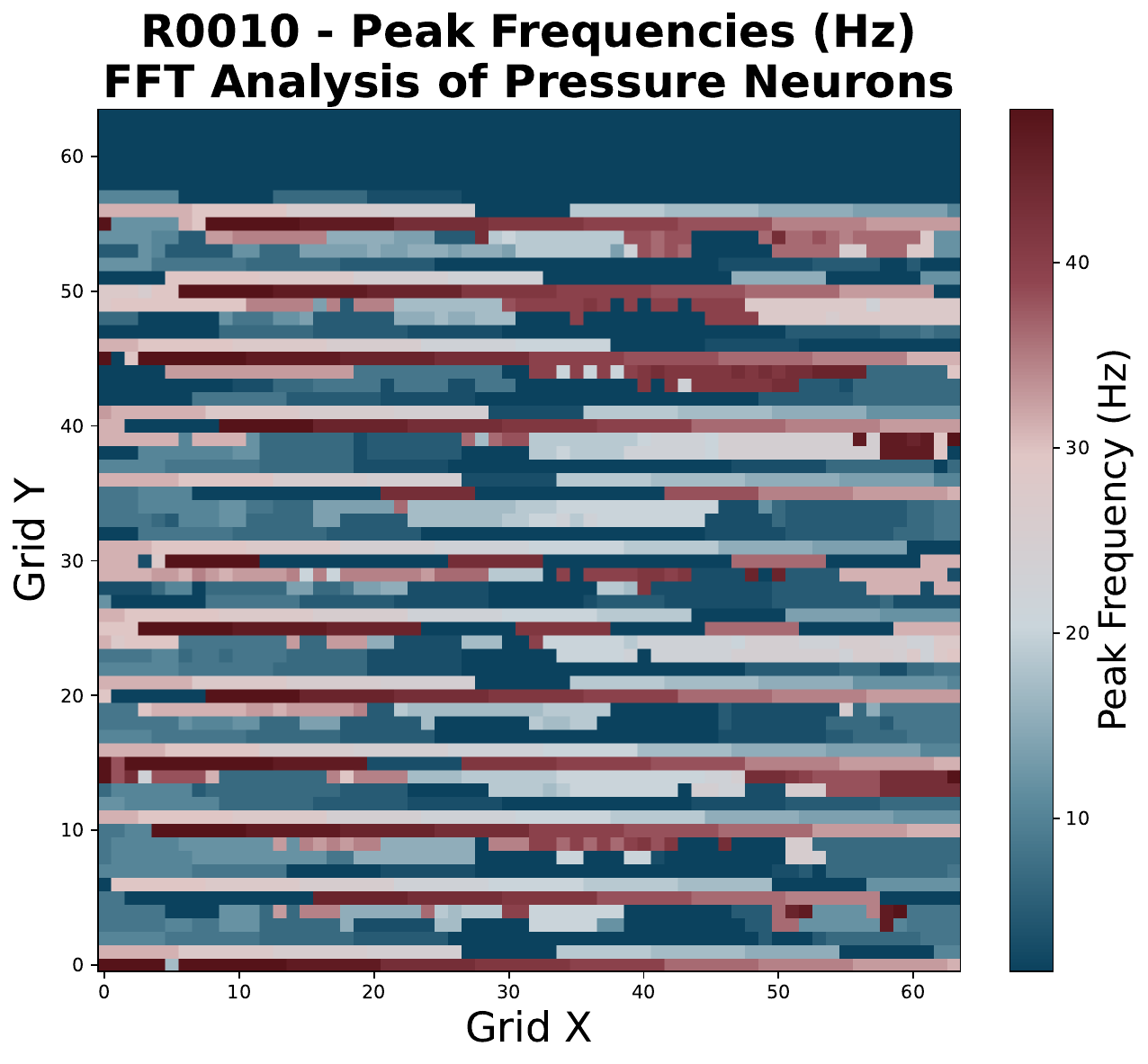}
        \caption{R0010}
    \end{subfigure}\hfill
    \begin{subfigure}{0.2\textwidth}
        \includegraphics[width=\linewidth,trim=0 0 0 60,clip]{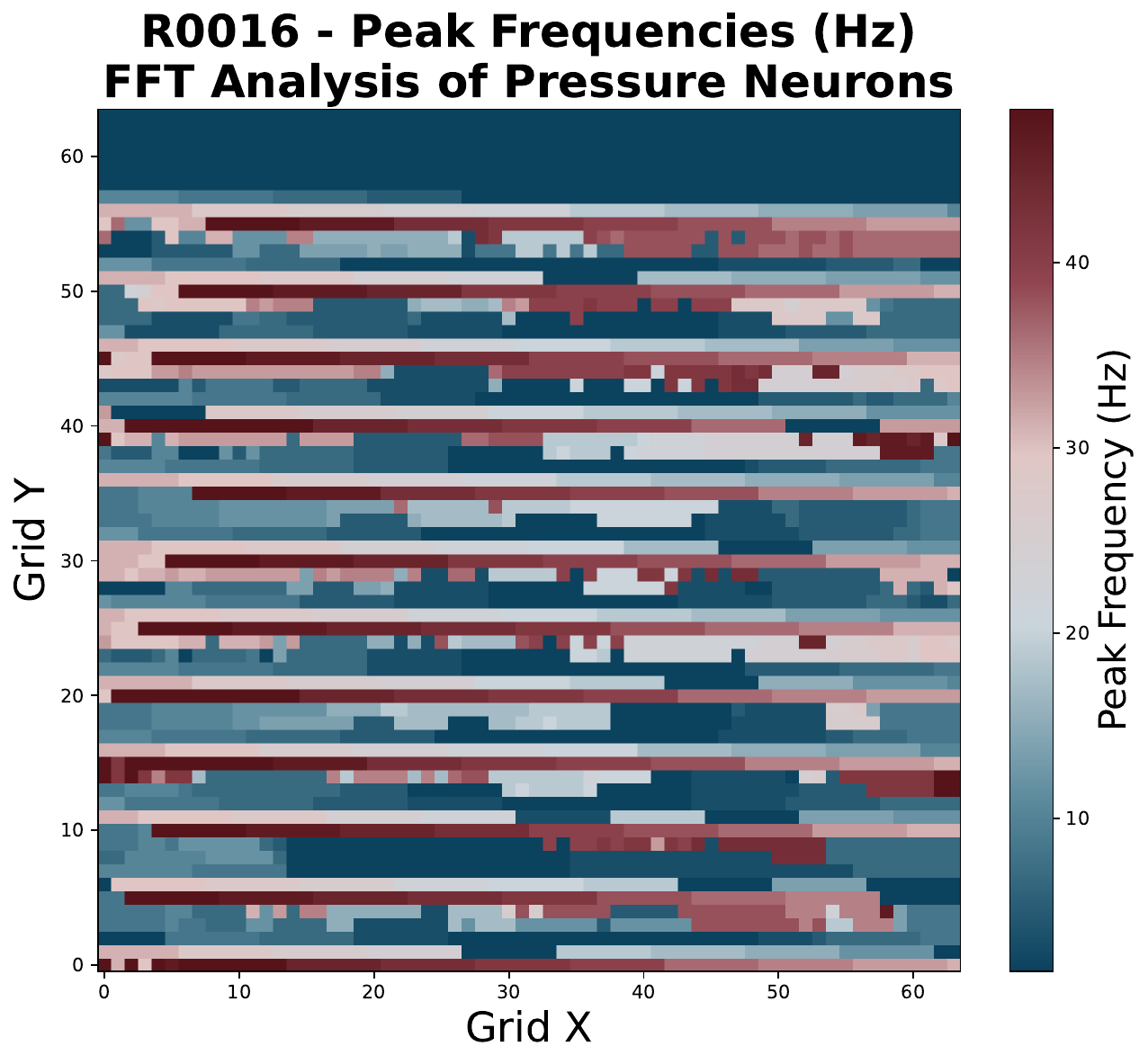}
        \caption{R0016}
    \end{subfigure}\hfill
    \begin{subfigure}{0.2\textwidth}
        \includegraphics[width=\linewidth,trim=0 0 0 60,clip]{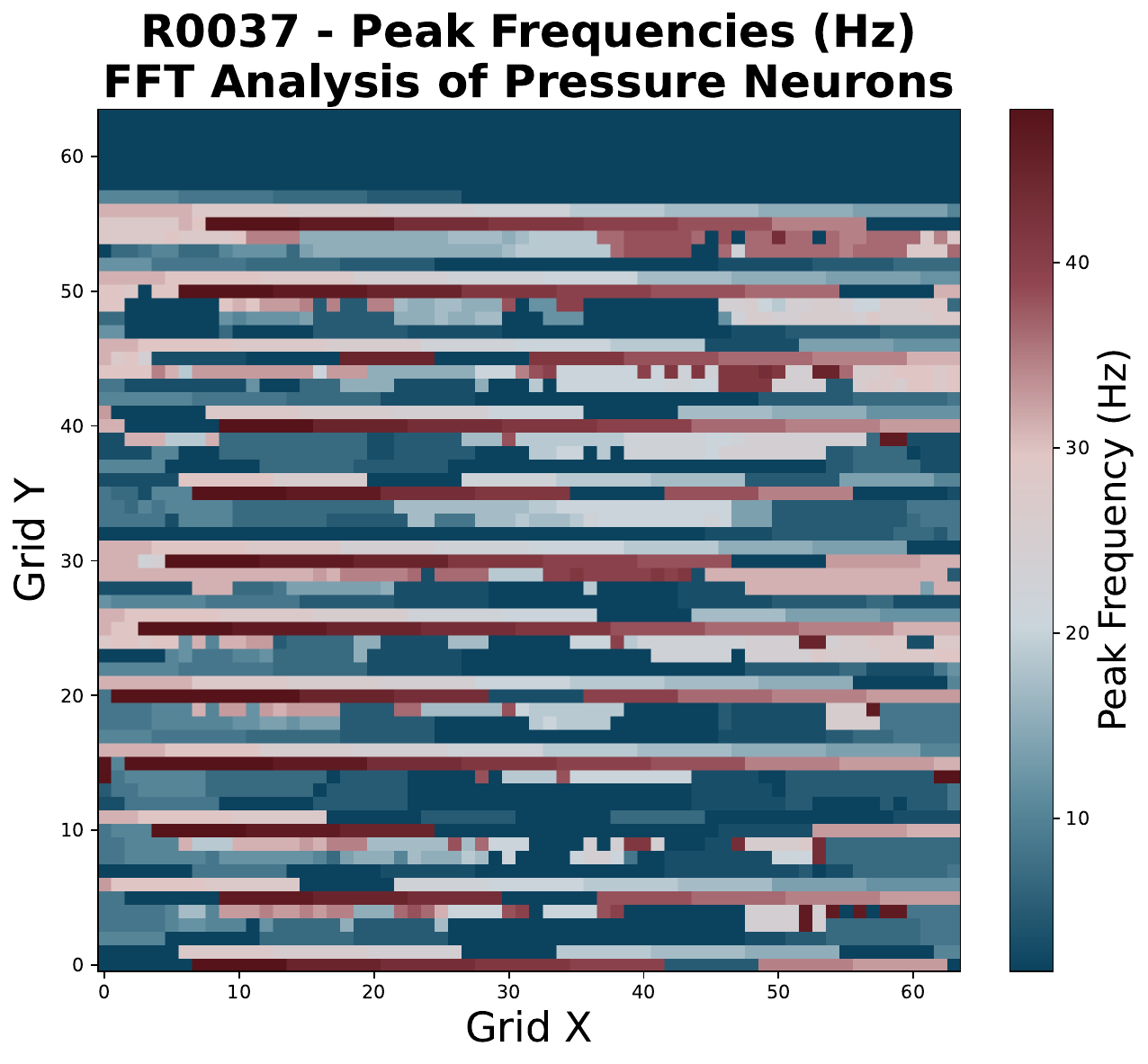}
        \caption{R0037}
    \end{subfigure}\hfill
    \begin{subfigure}{0.2\textwidth}
        \includegraphics[width=\linewidth,trim=0 0 0 60,clip]{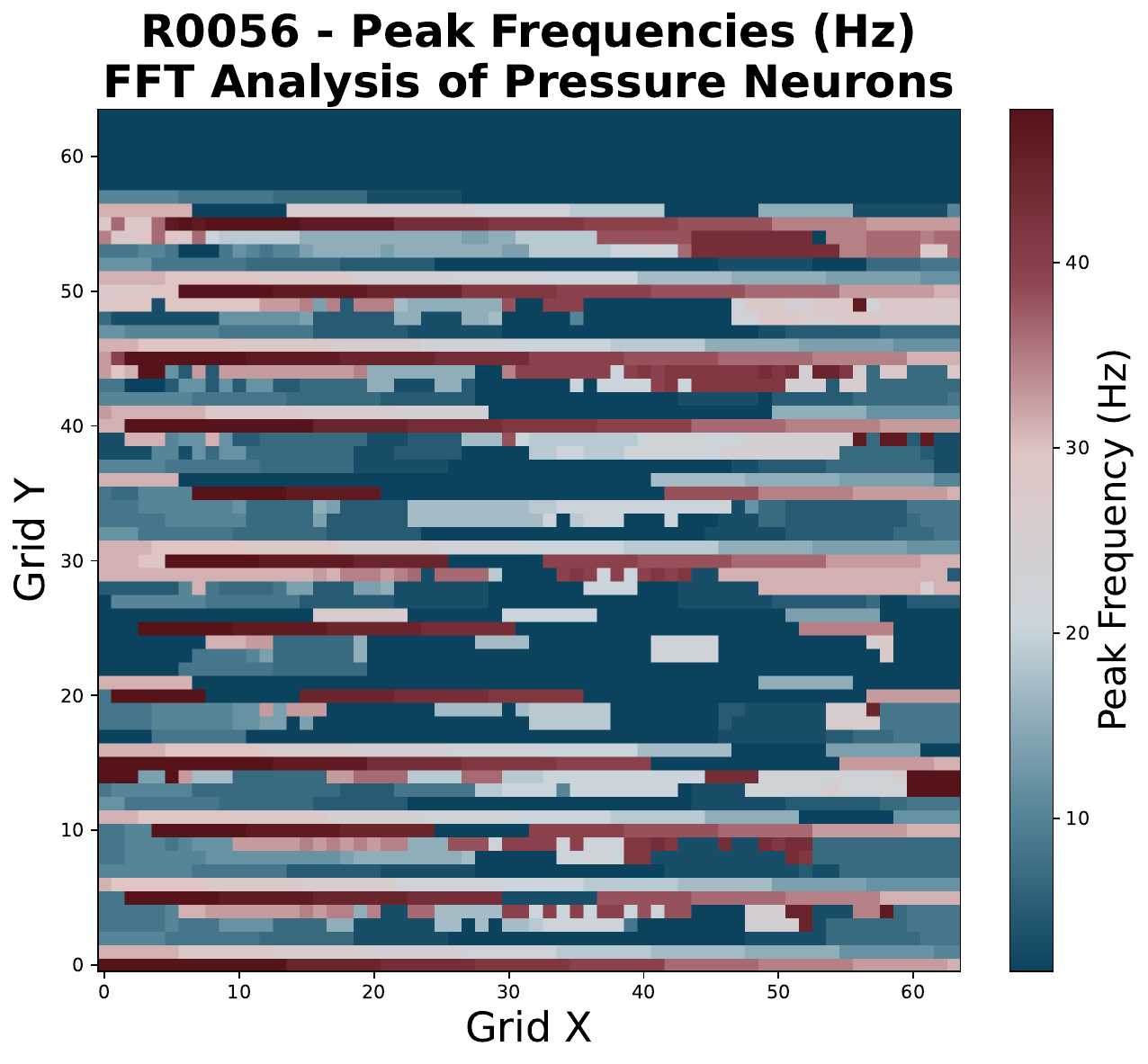}
        \caption{R0056}
    \end{subfigure}
    \caption{Frequency distribution across the \(p\)-field neurons obtained via FFT for different Examples (R0002, R0010, R0016, R0037, R0056). The results reveal band-specific activity in the theta (4--8\,Hz), alpha (8--12\,Hz), beta (13--30\,Hz), and gamma (30--50\,Hz) ranges.}
    \label{fig:FFT_figures_for_all_p_neurons}
\end{figure}

In the main paper (Section~\ref{subsec:FFT_analysis}), we presented FFT analyses for Examples R0016 and R0056 to
illustrate how BioWM reallocates oscillatory activity around drift onsets. To provide further
evidence, Fig.~\ref{fig:FFT_figures_for_all_p_neurons} shows additional FFT spectra for five
representative cases (R0002, R0010, R0016, R0037, R0056). These examples span a variety of
acoustic scenes, including novel sound events (R0002), subcategory-level drift (R0010), transient
pauses (R0016, R0037), and salient novel sources (R0056).

Across all examples, the dominant oscillatory activity of $p$-field neurons lies within the $\theta$ (4–8 Hz), $\alpha$ (8–12 Hz), $\beta$ (13–30 Hz), and low-$\gamma$ (30–50 Hz) bands, consistent with canonical neural regimes. Rather than being uniformly distributed, activity is organized in clustered regions that evolve after drift onsets. The figures display the activity of all $p$ neurons over entire segments; thus, direct saliency patterns are not immediately visible. What becomes evident instead is the periodic coupling structure along the Y-axis of the lattice, in line with the striped optimality predicted by Theorem~\ref{thm:striped-bragg}. This suggests that BioWM’s drift sensitivity arises from frequency-specific clustering combined with spatial coupling, rather than from broad or diffuse spectral fluctuations. Animated visualizations of these dynamics are included as supplementary material.




\end{document}